\documentclass[journal]{IEEEtran}

\usepackage{subfigure}
\usepackage{setspace}
\usepackage{multirow} 
\usepackage{algorithm}
\usepackage{algorithmic}
\usepackage{amsmath}
\usepackage{amssymb}
\usepackage{latexsym}
\usepackage{multirow}
\usepackage{epsfig}
\usepackage{graphics}
\usepackage{graphicx}
\usepackage{mathrsfs}
\usepackage{subfigure}
\usepackage{mathrsfs}
\usepackage{bbding}
\usepackage{cite}
\usepackage{color}
\usepackage{xcolor}
\usepackage{booktabs}
\usepackage{amssymb,mathrsfs,amsmath}
\newcommand{\bm}[1]{\mbox{\boldmath{$#1$}}}
\usepackage{flushend}
\usepackage{stfloats}
\usepackage{hyperref}
\usepackage{amsthm}
\theoremstyle{plain}

\newtheorem{prop}{Proposition}
\newtheorem{lem}{Lemma}

\allowdisplaybreaks [4]

\begin{document}	
	
\title{\fontsize{23pt}{30pt}\selectfont Movable Antennas Enabled Wireless-Powered NOMA: Continuous and Discrete Positioning Designs}
\author{Ying~Gao, Qingqing~Wu,~\IEEEmembership{Senior Member,~IEEE}, and Wen~Chen,~\IEEEmembership{Senior Member,~IEEE}
\thanks{The authors are with the Department of Electronic Engineering, Shanghai Jiao Tong University, Shanghai 201210, China (e-mail: yinggao@sjtu.edu.cn; qingqingwu@sjtu.edu.cn; wenchen@sjtu.edu.cn).}}

\maketitle 

\begin{abstract}
	This paper investigates a movable antenna (MA)-enabled wireless-powered communication network (WPCN), where multiple wireless devices (WDs) first harvest energy from the downlink signal broadcast by a hybrid access point (HAP) and then transmit information in the uplink using non-orthogonal multiple access. Unlike conventional WPCNs with fixed-position antennas (FPAs), this MA-enabled WPCN allows the MAs at the HAP and the WDs to adjust their positions twice: once before downlink wireless power transfer and once before uplink wireless information transmission. Our goal is to maximize the system sum throughput by jointly optimizing the MA positions, the time allocation, and the uplink power allocation. Considering the characteristics of antenna movement, we explore both continuous and discrete positioning designs, which, after formulation, are found to be non-convex optimization problems. Before tackling these problems, we rigorously prove that using identical MA positions for both downlink and uplink is the optimal strategy in both scenarios, thereby greatly simplifying the problems and enabling easier practical implementation of the system. We then propose alternating optimization-based algorithms to obtain suboptimal solutions for the resulting simplified problems. Simulation results show that: 1) the proposed continuous MA scheme can enhance the sum throughput by up to 395.71\% compared to the benchmark with FPAs, even when additional compensation transmission time is provided to the latter; 2) a step size of one-quarter wavelength for the MA motion driver is generally sufficient for the proposed discrete MA scheme to achieve over 80\% of the sum throughput performance of the continuous MA scheme; 3) when each moving region is large enough to include multiple optimal positions for the continuous MA scheme, the discrete MA scheme can achieve comparable sum throughput without requiring an excessively small step size. 
\end{abstract}


\vspace{-1mm}
\section{Introduction}\label{Sec:intro}
Harvesting energy from the environment provides a cost-effective and virtually unlimited power source for wireless devices (WDs), presenting a greener and more practical alternative to conventional battery-powered methods \cite{2011_Sujesha_EH_survey}. Among various renewable energy options, including solar and wind, radio-frequency signals stand out for energy harvesting (EH) due to their widespread availability and ease of control. A typical application of radio-frequency-based EH can be found in wireless-powered communication networks (WPCNs). Two distinct lines of research have emerged based on whether the energy node and the information access point (AP) are geographically co-located or separated. In the co-located scenario, the well-known harvest-then-transmit protocol was proposed in \cite{2014_Hyungsik_WPCN} for a multiuser single-input single-output (SISO) WPCN using a half-duplex hybrid AP (HAP). Building on this, reference \cite{2014_Hyungsik_full-duplex} introduced a full-duplex HAP, allowing simultaneous energy transfer and data reception through self-interference cancellation. Furthermore, the authors of \cite{2014_Liang_WPCN} extended the single-antenna HAP in \cite{2014_Hyungsik_WPCN} to a multi-antenna configuration, enhancing energy transmission efficiency through beamforming techniques. Beyond these single-cell WPCN studies, reference \cite{2018_Hanjin_WPCN} explored wireless-powered communications in a multiuser SISO interference channel. On the other hand, in the separated scenario, the authors of \cite{2016_Feng_WPCN} considered a three-node WPCN under the harvest-then-transmit protocol. In contrast, reference \cite{2016_Xun_WPCN} explored the scenario where wireless power transfer (WPT) and wireless information transmission (WIT) occurred over orthogonal sub-channels. The model in \cite{2014_Qian_WPCN} extended \cite{2016_Xun_WPCN} by incorporating multiple WDs and adopting multiple antennas at the energy node. 

Despite these theoretical developments, practical WPCNs still suffer from significant limitations, mainly due to the low efficiency of WPT and WIT over long distances. To address this, massive multiple-input multiple-output (MIMO) technology has been proposed as a breakthrough, utilizing a large number of extra antennas to focus energy into ever-smaller regions of space \cite{2015_Gang_massive}. However, implementing massive MIMO systems requires numerous parallel radio-frequency chains, resulting in a substantial increase in hardware costs and energy consumption. Antenna selection offers an effective way to reduce the number of radio-frequency chains, allowing massive MIMO systems to capture much of the channel capacity by selecting a small subset of antennas with favorable channels from a larger pool \cite{2004_Molisch_AS,2004_Sanayei_AS}. Nonetheless, as the number of candidate antennas increases, the computational costs of channel estimation and antenna selection algorithms also grow. Intelligent reflecting surfaces (IRSs) have also been proposed as a cost-effective solution to enhance WPCN performance, leveraging their ability to reconfigure wireless channels by reflecting incident signals in desired directions \cite{2019_Qingqing_Joint,2022_Qingqing_WEIT_overview,2024_Ying_WPCN_IRS,2025_Qingqing_Deployment}. However, integrating IRSs as third-party devices into communication systems adds complexity and can potentially impact overall reliability. Furthermore, whether using massive MIMO, antenna selection, or IRS technologies, the fixed positions of transmit/receive antennas limit the ability to fully exploit channel variations in the continuous spatial field. 

Recently, movable antennas (MAs) \cite{2023_Lipeng_Modeling,2023_Wenyan_MIMO,2023_Lipeng_overview}, also termed fluid antennas \cite{2021_Wong_fluid,2020_Wong_Fluid}, have garnered considerable academic attention as a promising solution to overcome the inherent limitations of fixed-position antenna (FPA)-based systems. While MAs have long existed in antenna technology, systematic research into their wireless communication applications has only recently emerged. MAs are connected to radio-frequency chains via flexible cables, with positions dynamically adjusted using controllers like stepper motors or servos. Unlike conventional FPA systems, MA-enabled systems can reposition transmit/receive antennas, reconfiguring channel conditions to fully exploit spatial diversity, mitigate interference, and improve spatial multiplexing gains \cite{2023_Lipeng_uplink}. These advantages have driven several research efforts on MA-assisted communication systems. For instance, reference \cite{2023_Lipeng_Modeling} introduced a mechanical MA architecture and a field-response based channel model for single-MA systems, examining conditions for alignment with various channel types. The study also analyzed the signal-to-noise (SNR) ratio gain of a single receive MA compared to its FPA counterpart, showing that performance improvements depend heavily on the number of channel paths and the MA's spatial movement area. In \cite{2023_Wenyan_MIMO}, MA-enabled MIMO systems were investigated, where the positions of transmit and receive MAs were jointly optimized along with the transmit covariance matrix to maximize the channel capacity. Additionally, other studies have explored joint MA positioning and resource allocation strategies to improve data rates \cite{2024_Nian_MA,2024_Biqian_MA,2024_Yichi_hybrid}, enhance user fairness \cite{2023_Zhenyu_uplink,2024_Ying_MA}, suppress interference \cite{2023_Lipeng_null}, conserve power \cite{2023_Lipeng_uplink,2024_Haoran_MA,2024_Honghao_IFC,2023_Guojie_gradient}, strengthen physical security \cite{2024_Zhenqiao_secure,2024_Guojie_secure_lett,2024_Guojie_secure_long}, facilitate spectrum sharing \cite{2024_Weidong_MA_spectrumsharing}, and achieve wide-beam coverage \cite{2025_Weidong_MA_widebeam}.
All of the above works assumed continuous MA positioning within a given area. While this offers maximum flexibility and characterizes the performance limit, implementation may be challenging due to the discrete movement constraints of practical stepper motors. Given this, some studies have modeled MA motion as discrete steps and investigated discrete positioning designs for objectives such as transmit power minimization \cite{2023_Yifei_discrete} and received signal power maximization \cite{2024_Weidong_graph}. 

Although the performance of MAs in wireless communications has been explored across various system setups, research on MA-aided (or fluid antenna-aided) WPCNs is still in its early stages. Among these studies, reference \cite{2024_Xiazhi_FAS_WPCN} examined a scenario where a transmitter equipped with a fluid antenna is powered by an energy node and then uses the harvested energy to send data to a single receiver. The study assumed that one of several switchable fluid antenna ports is selected for both WPT and WIT, and derived analytical and asymptotic expressions of the outage probability to evaluate the system performance. Reference \cite{2024_Ghadi_FAS_WPCN} then extended this setup to two receivers, where the fluid antenna-equipped devices are the receivers rather than the transmitter. However, these studies are limited to single- or two-receiver configurations, which are not ideal for broader applications, and do not fully leverage MA/fluid antenna technology across all involved devices for further performance improvements. Moreover, since MA positions can be dynamically adjusted, a fundamental question remains unanswered in MA-aided WPCNs: \emph{is using different MA positions for downlink WPT and uplink WIT beneficial for maximizing system sum throughput when the energy node and the information AP are geographically co-located}? This question stems from the fact that downlink WPT and uplink WIT occur in different time slots and have distinct objectives, but involve the same devices. Thus, it is unknown whether employing different MA positions for these two phases is the optimal strategy in this context.   

\begin{figure}[!t]
	\centering
	\includegraphics[width=0.49\textwidth]{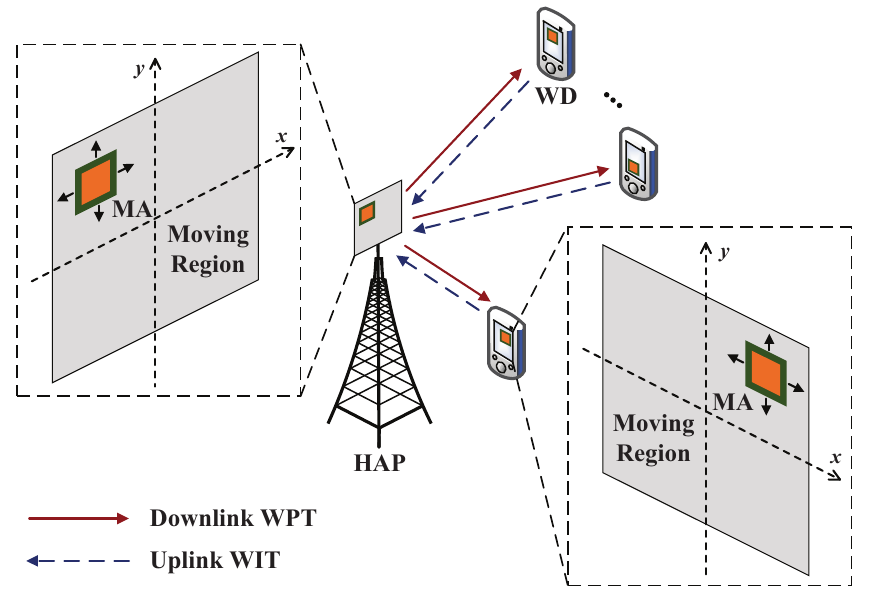}
	\caption{Illustration of an MA-enabled WPCN.} \label{Fig:system_model}
	\vspace{-1mm}
\end{figure}

In addition, while some studies, such as \cite{2023_Yifei_discrete} and \cite{2024_Weidong_graph}, have investigated discrete MA positioning designs for various objectives, they do not evaluate the performance gap between continuous and discrete positioning or clarify the numerical thresholds or conditions under which discrete positioning can effectively capture most of the performance benefits of continuous positioning. Naturally, for discrete MA positioning, a smaller step size of the MA motion driver leads to an increased number of candidate positions. As this number approaches infinity, discrete positioning can be considered equivalent to continuous positioning. However, a tiny step size imposes high demands on hardware, resulting in significantly increased costs. Furthermore, the complexity of position selection rises with the number of candidate options. Therefore, determining the appropriate step size that ensures a satisfactory level of performance is of considerable engineering significance. In light of this concern, from a theoretical perspective, if the MA moving region is sufficiently large, discrete positioning can achieve performance comparable to continuous positioning, even with a step size that is not extremely small. This is because the maximum performance of continuous positioning can be achieved within a finite MA moving region, and the superimposed power of multiple channel paths exhibits a periodic nature in the receive region due to the existence of the cosine function \cite{2023_Lipeng_Modeling}. In other words, multiple optimal positions exist for continuous positioning design to reach maximum performance when a sufficiently large MA moving region is available. As a consequence, there is a high likelihood that the candidate options in discrete positioning design include some of these optimal positions, even if the total number of candidates is not considerably large. 

Motivated by the above discussions, we investigate an MA-enabled WPCN, comprising multiple WDs and a HAP, each equipped with an MA, as shown in Fig. \ref{Fig:system_model}. This system follows the typical harvest-then-transmit protocol \cite{2014_Hyungsik_WPCN} and utilizes non-orthogonal multiple access (NOMA) for uplink WIT. In particular, the MA at each device can adjust its position prior to both downlink WPT and uplink WIT, offering flexibility that is absent in conventional FPA-based WPCNs. We aim to maximize the system sum throughput through the joint optimization of MA positions, time allocation, and uplink power allocation. Given the characteristics of antenna movement, we study both continuous and discrete positioning designs, formulated as non-convex and mixed-integer non-convex optimization problems, respectively. Our main contributions are summarized as follows. 
\begin{itemize}
	\item  For the optimization problem involving continuous positioning, we first reveal that the optimum is achieved when the MA positions are identical for both downlink WPT and uplink WIT. This result not only simplifies the problem but also makes the system easier to implement in practice.  Building on this, we propose an iterative algorithm based on alternating optimization (AO) to solve the resulting simplified problem. While each position variable is not explicitly exposed in the objective function of its corresponding subproblem, which involves the fourth power of the absolute value, we efficiently solve these subproblems using the successive convex approximation (SCA) technique.  
	\item For the optimization problem involving discrete positioning, we similarly prove that using identical MA positions for both downlink and uplink is optimal. We then apply the AO method to solve the resulting simplified problem by dividing the optimization variables into three blocks. Each MA position-related binary optimization variable can be optimally determined in its corresponding subproblem using an exhaustive search. In particular, by leveraging the special structure of these subproblems, we streamline the exhaustive search process, significantly reducing the complexity of determining the optimal solutions for these binary variables to a more acceptable level. 
	\item Numerical results demonstrate that the proposed continuous MA scheme achieves notable improvements, enhancing the sum throughput by up to 395.71\% compared to the FPA benchmark, even when the latter is given additional transmission time for compensation. Moreover, for the discrete MA scheme, a step size of one-quarter wavelength for the MA motion driver is generally sufficient to attain over 80\% of the performance achieved by the continuous MA scheme in terms of sum throughput. Additionally, if each moving region is sufficiently large to encompass multiple optimal positions for the continuous MA scheme, the discrete MA scheme can deliver comparable sum throughput without demanding a tiny step size. This finding supports our theoretical analysis presented in the previous paragraph. 
\end{itemize}

The rest of this paper is structured as follows. In Section \ref{Sec:model_and_formulation}, we describe the system model and introduce two problem formulations for an MA-enabled WPCN, focusing on antenna position optimization in both continuous regions and discrete sets. Sections \ref{Sec:P1_solu} and \ref{Sec:P2_solu} detail the proposed algorithms for these two problems. Section \ref{Sec:simu} presents numerical simulations to assess the effectiveness of our proposed algorithms. Lastly, the conclusions are drawn in Section \ref{Sec:conclu}. 

\emph{Notations:} Let $\mathbb{C}$ denote the complex space, and $\mathbb{C}^{M \times N}$ represent the space of $M \times N$ matrices with complex-valued entries. For any complex number $x$, its modulus and phase are represented by $\left|x\right|$ and $\arg(x)$, respectively. For a vector $\mathbf{x}$, $\left\|\mathbf{x}\right\|_1$ denotes its $l_1$ norm, while $\left[\mathbf{x}\right]_i$ refers to its $i$-th element. For a matrix $\mathbf{X}$ of any size, $\left\|\mathbf{X}\right\|_2$ is the spectral norm, $\left\|\mathbf{X}\right\|_F$ is the Frobenius norm, and $\left[\mathbf{X}\right]_{i,j}$ represents the element located at the $i$-th row and $j$-th column. When comparing two square matrices $\mathbf{A}_1$ and $\mathbf{A}_2$, the notation $\mathbf{A}_1 \succeq \mathbf{A}_2$ implies that $\mathbf{A}_1 - \mathbf{A}_2$ is positive semidefinite. The identity matrix is denoted by $\mathbf{I}$, with its dimensions inferred from the context. The Hermitian (conjugate) transpose is denoted by $(\cdot)^H$, and $\mathbb{E}(\cdot)$ stands for the expectation operator. The operation ${\rm diag}(\cdot)$ converts a vector into a diagonal matrix. The notation $\mathcal{CN}(0, \sigma^2)$ represents a complex Gaussian distribution with zero mean and variance $\sigma^2$. Finally, $\jmath \triangleq \sqrt{-1}$ is used to denote the imaginary unit.

\begin{figure}[!t]
	\centering
	\includegraphics[width=0.48\textwidth]{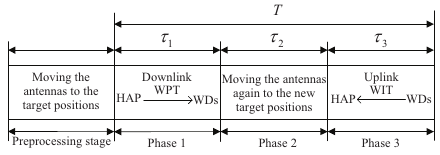}
	\caption{Illustration of the transmission protocol.} \label{Fig:trans_protocol}
\end{figure}

\section{System Model and Problem Formulation}\label{Sec:model_and_formulation}
As illustrated in Fig. \ref{Fig:system_model}, we consider an MA-enabled WPCN consisting of a single-MA HAP and $K$ single-MA WDs, indexed by $\mathcal K \triangleq \{1,\cdots,K\}$. 
The MAs are connected to radio-frequency chains via flexible cables. This configuration allows the position of the MA at the HAP to be adjusted within a given two-dimensional (2D) region $\mathcal C_\omega$ and the position of the MA at WD $k$ to be adjusted within a given 2D region $\mathcal C_k$.\footnote{We focus on a 2D movement region in this work, primarily due to the increased mechanical complexity and energy consumption that would be incurred in the three-dimensional (3D) case. Nevertheless, the proposed algorithmic framework is inherently extensible to the 3D case through appropriate generalization of the MA position parameterization and corresponding modifications to the channel model.} The reference points in the regions $\mathcal C_\omega$ and $\mathcal C_k$ are denoted by $\boldsymbol \omega_0 = [0,0]^T$ and $\boldsymbol u_{k,0} = [0,0]^T$, respectively. These reference points also serve as the initial positions of the HAP- and WD-side MAs. Both $\mathcal C_\omega$ and $\mathcal C_k$ are assumed to be square regions of size $A\times A$. 
We assume that the far-field condition is satisfied between the HAP and the WDs. Thus, altering the MA positions only affects the phase of the complex coefficient for each channel path component. The angle of departure (AoD), the angle of arrival (AoA), and the amplitude of the complex coefficient remain unchanged \cite{2023_Lipeng_Modeling}. Let $L_k$ and $\tilde L_k$ represent the total number of transmit and receive channel paths from the HAP to WD $k$, respectively. For the $i$-th transmit path to WD $k$, the elevation and azimuth AoDs are $\theta_k^i \in \left[0, \pi\right]$ and $\phi_k^i\in \left[0,\pi\right]$, respectively. For the $j$-th receive path to WD $k$, the elevation and azimuth AoAs are $\tilde \theta_k^j \in \left[0,\pi\right]$ and $\tilde \phi_k^j \in \left[0,\pi\right]$, respectively. Furthermore, the uplink-downlink channel reciprocity is assumed. To quantify the performance limit, we assume that perfect channel state information (CSI) is available at the HAP.\footnote{MA channel acquisition falls into two main categories: model-based and model-free methods. Model-based approaches rely on the field-response channel model and often adopt techniques such as compressed sensing \cite{2023_Wenyan_estimation} or tensor decomposition \cite{2024_Ruoyu_channel}. In contrast, model-free methods perform channel measurements at selected positions and reconstruct unmeasured channels either by assuming spatial correlation with nearby samples \cite{2023_Christodoulos_channel}, by applying Bayesian regression \cite{2024_Zijian_channel}, or by leveraging machine learning techniques \cite{2024_Shuyan_channel}.}


The transmission protocol for the MA-enabled WPCN is illustrated in Fig. \ref{Fig:trans_protocol}. Like conventional FPA-based WPCNs, the MA-enabled WPCN follows the typical harvest-then-transmit protocol \cite{2014_Hyungsik_WPCN}. However, it differentiates itself by permitting the antennas to change their positions twice: once before downlink WPT and once before uplink WIT. In alignment with existing works \cite{2024_Xiazhi_FAS_WPCN,2023_Xiao_FAS,2024_Pengcheng_MA}, the first movement of the MAs is treated as a preprocessing stage, and the time spent on this movement is not included in the total transmission time $T$. The entire transmission process is divided into three phases: one for downlink WPT, one for the second movement of the MAs, and one for uplink WIT, as shown in Fig. \ref{Fig:trans_protocol}. We define $\tau_1$, $\tau_2$, and $\tau_3$ as the time durations of these three phases. In addition, this study targets practical scenarios characterized by slowly varying wireless environments, such as indoor Machine-Type Communication (MTC) systems, where devices are typically stationary and propagation conditions remain stable over extended periods. In such scenarios, the channel at each location within the designated moving regions can be regarded as effectively constant throughout the preprocessing phase and the subsequent transmission period.

The following study is based on two movement patterns of the MAs: continuous and discrete. The former is employed to quantify the performance limit, while the latter is more consistent with practical implementation. By considering both schemes, we evaluate their performance gap and extract valuable engineering insights. 

\subsection{Continuous Antenna Positioning}
In this scenario, the MAs can move freely and continuously within given regions. It is assumed that before the start of downlink WPT, the MA at the HAP and the MA at WD $k$ have been moved to the positions $\bm \omega_1 = \left[x_{\omega,1}, y_{\omega,1}\right]^T\in\mathcal C_\omega$ and $\boldsymbol u_{k,1} = \left[x_{k,1}, y_{k,1}\right]^T\in\mathcal C_k$, respectively.\footnote{This work focuses on evaluating the communication performance of MA-aided WPCNs under the assumption of successful antenna positioning. As described in \cite{2023_Lipeng_overview}, a typical MA-mounted transmitter/receiver comprises two distinct modules: a communication module and an antenna positioning module. Since these modules serve distinct functions, they do not necessarily share the same power supply. While antenna movement incurs mechanical energy consumption, this is not modeled here, as the movement can be powered separately and occurs infrequently in applications such as indoor MTC. A joint analysis of mechanical and communication energy consumption is a promising direction for future work.} 
Then, during phase 1, the HAP broadcasts an energy signal to all the WDs with a constant transmit power $P_{\rm A}$ for a duration of $\tau_1$. 
By adopting the linear EH model and ignoring the negligible noise power, the energy harvested by WD $k$ in the downlink can be expressed as $E_k = \zeta P_{\rm A}\left|h_{k,1}(\boldsymbol \omega_1, \boldsymbol u_{k,1})\right|^2\tau_1$, where $\zeta$ stands for the constant energy conversion efficiency for each WD, and $h_{k,1}(\boldsymbol \omega_1, \boldsymbol u_{k,1})$ denotes the downlink channel from the HAP to WD $k$, which is determined by the propagation environment and the positions $\boldsymbol \omega_1$ and $\boldsymbol u_{k,1}$. Specifically,
\begin{align} 
	h_{k,1}(\boldsymbol \omega_1, \boldsymbol u_{k,1}) = \bm f_k(\boldsymbol u_{k,1})^H\mathbf \Sigma_k\bm g_k(\boldsymbol \omega_1), 
\end{align}
where $\mathbf \Sigma_k \in \mathbb C^{\tilde L_k\times L_k}$ represents the path-response matrix characterizing the responses between all the transmit and receive channel paths from $\boldsymbol \omega_0$ to $\boldsymbol u_{k,0}$, and $\bm g_k(\boldsymbol \omega_1)$ and $\bm f_k(\boldsymbol u_{k,1})$ denote the transmit and receive field-response vectors for the channel from the HAP to WD $k$, respectively, given by \cite{2023_Wenyan_MIMO}
\begin{subequations}\label{eqv:FRV}
\begin{align}
	\bm g_k(\boldsymbol \omega_1) & = \left[e^{\jmath\frac{2\pi}{\lambda}\boldsymbol \omega_1^T\boldsymbol a_k^1}, e^{\jmath\frac{2\pi}{\lambda}\boldsymbol \omega_1^T\boldsymbol a_k^2}, \cdots, e^{\jmath\frac{2\pi}{\lambda}\boldsymbol \omega_1^T\boldsymbol a_k^{L_k}}\right]^T,\\
	\bm f_k(\boldsymbol u_{k,1}) & = \left[e^{\jmath\frac{2\pi}{\lambda}\boldsymbol u_{k,1}^T\tilde{\boldsymbol a}_k^1}, e^{\jmath\frac{2\pi}{\lambda}\boldsymbol u_{k,1}^T\tilde{\boldsymbol a}_k^2}, \cdots, e^{\jmath\frac{2\pi}{\lambda}\boldsymbol u_{k,1}^T\tilde{\boldsymbol a}_k^{\tilde L_k}}\right]^T, 
\end{align}
\end{subequations}
with $\bm a_k^i\triangleq \left[\sin\theta_k^i\cos\phi_k^i, \cos\theta_k^i\right]^T$, $i\in\{1,\cdots,L_k\}$ and $\tilde{\bm a}_k^j \triangleq \left[\sin\tilde{\theta}_k^j\cos\tilde{\phi}_k^j, \cos\tilde{\theta}_k^j\right]^T$, $j\in\{1,\cdots,\tilde L_k\}$. 

In the subsequent phase 2, the MAs are moved using step motors along slide tracks (for details on the hardware architecture, please refer to \cite[Fig. 2]{2023_Lipeng_overview}). Without loss of generality, we assume that all the MAs start moving simultaneously at the same speed, denoted by $v$, measured in meter/second (m/s). Specifically, the MA at the HAP moves from the position $\boldsymbol \omega_1 = \left[x_{w,1}, y_{w,1}\right]^T\in\mathcal C_\omega$ to the position $\boldsymbol \omega_2 = \left[x_{w,2}, y_{w,2}\right]^T\in\mathcal C_\omega$, while the MA at WD $k$ moves from the position $\boldsymbol u_{k,1} = \left[x_{k,1}, y_{k,1}\right]^T\in\mathcal C_k$ to the position $\boldsymbol u_{k,2} = \left[x_{k,2}, y_{k,2}\right]^T\in\mathcal C_k$. Considering that the movement of the MAs involves both the x-axis and y-axis directions, the time required for all the MAs to complete their movement is:
\begin{align}\label{eq:tau_2}
	\tau_2 = \max\left(\frac{\left\|\boldsymbol \omega_2 - \boldsymbol \omega_1\right\|_1}{v}, \max_{k\in\mathcal K}\frac{\left\| \boldsymbol u_{k,2} - \boldsymbol u_{k,1}\right\|_1}{v}\right).
\end{align}

In the final uplink WIT phase, all the WDs utilize their harvested energy to transmit their respective information signals to the HAP simultaneously using NOMA. The successive interference cancellation technique \cite{2014_Zhiguo_NOMA} is employed at the HAP to eliminate multiuser interference. Let $p_k$ and $\pi_k$ denote the transmit power and decoding order of WD $k$. Then, the achievable throughput of WD $k$ in bits/Hz can be expressed as \cite{2014_Mohammed_NOMA} 
\begin{align}
	 R_k^{\rm cont} \!=\! \tau_3\log_2\left(1 \!+\! \frac{p_k\left|h_{k,2}(\boldsymbol \omega_2, \boldsymbol u_{k,2})\right|^2}{\sum_{\pi_{k'} > \pi_k}p_{k'}\left|h_{k',2}(\boldsymbol \omega_2, \boldsymbol u_{k',2})\right|^2 + \sigma^2}\right), 
\end{align}
where $\sigma^2$ is the additive white Gaussian noise power at the HAP. Moreover, $h_{k,2}(\boldsymbol \omega_2, \boldsymbol u_{k,2})$ represents the uplink channel from WD $k$ to the HAP, given by 
\begin{align}
	h_{k,2}(\boldsymbol \omega_2, \boldsymbol u_{k,2}) & = \bm g_k(\boldsymbol \omega_2)^H\mathbf \Sigma_k^H\bm f_k(\boldsymbol u_{k,2}) \nonumber\\ 
	& = \left(\bm f_k(\boldsymbol u_{k,2})^H\mathbf \Sigma_k\bm g_k(\boldsymbol \omega_2)\right)^H, 
\end{align}
with 
\begin{subequations}
      \begin{align}
      	\bm f_k(\boldsymbol u_{k,2}) & = \left[e^{\jmath\frac{2\pi}{\lambda}\boldsymbol u_{k,2}^T\tilde{\boldsymbol a}_k^1}, e^{\jmath\frac{2\pi}{\lambda}\boldsymbol u_{k,2}^T\tilde{\boldsymbol a}_k^2}, \cdots, e^{\jmath\frac{2\pi}{\lambda}\boldsymbol u_{k,2}^T\tilde{\boldsymbol a}_k^{\tilde L_k}}\right]^T, \\
      	\bm g_k(\boldsymbol \omega_2) & = \left[e^{\jmath\frac{2\pi}{\lambda}\boldsymbol \omega_2^T\boldsymbol a_k^1}, e^{\jmath\frac{2\pi}{\lambda}\boldsymbol \omega_2^T\boldsymbol a_k^2}, \cdots, e^{\jmath\frac{2\pi}{\lambda}\boldsymbol \omega_2^T\boldsymbol a_k^{L_k}}\right]^T
      \end{align}
\end{subequations}
being the transmit and receive field-response vectors, respectively. As a result, the system sum throughput is given by 
\begin{align}
	R_{\rm sum}^{\rm cont} & = \sum_{k=1}^K R_k^{\rm cont} \nonumber\\
	& = \tau_3\log_2\left(1 + \sum_{k=1}^K\frac{p_k\left|h_{k,2}(\boldsymbol \omega_2, \boldsymbol u_{k,2})\right|^2}{\sigma^2}\right), 
\end{align}
which is independent of the decoding order. 

\vspace{-2mm}
\subsection{Discrete Antenna Positioning}
In this scenario, MA motion is modeled as discrete steps, with a step size of $d$ assumed for each MA without loss of generality. The number of candidate discrete positions for each MA is then determined by $d$ and the size of its moving region. We assume there are $N_\omega$ candidate discrete positions for the MA at the HAP, denoted as $\mathbf p_1, \cdots, \mathbf p_{N_\omega}$, and $N_{u,k}$ candidate discrete positions for the MA at WD $k$, denoted as $\mathbf q_{k,1}, \cdots, \mathbf q_{k,N_{u,k}}$. It is further assumed that before the start of phase 1, the MA at the HAP is moved to the position $\sum_{m=1}^{N_\omega}s_1^m\mathbf p_m$, and the MA at WD $k$ is moved to the position $\sum_{n=1}^{N_{u,k}}t_{k,1}^n\mathbf q_{k,n}$. Here, $s_1^m \in \{0,1\}$ with $\sum_{m=1}^{N_\omega}s_1^m = 1$ is a binary variable indicating the position of the HAP's MA in phase 1. Similarly, $t_{k,1}^n \in \{0,1\}$ with $\sum_{n=1}^{N_{u,k}}t_{k,1}^n = 1$ indicates the position of the $k$-th MD's MA in phase 1. Then, the energy harvested by WD $k$ during phase 1 is given by $E_k = \zeta P_{\rm A}\left|h_{k,1}(\{s_1^m\},\{t_{k,1}^n\})\right|^2\tau_1$,  
where 
\begin{align}
& h_{k,1}(\{s_1^m\},\{t_{k,1}^n\}) \nonumber\\
& = \left(\bm f_k\left(\sum_{n=1}^{N_{u,k}}t_{k,1}^n\mathbf q_{k,n}\right)\right)^H\mathbf \Sigma_k\left(\bm g_k\left( \sum_{m=1}^{N_\omega}s_1^m\mathbf p_m\right)\right) \nonumber\\
& = \left(\sum_{n=1}^{N_{u,k}}t_{k,1}^n\bm f_k(\mathbf q_{k,n})\right)^H\mathbf \Sigma_k\left( \sum_{m=1}^{N_\omega}s_1^m\bm g_k(\mathbf p_m)\right)
\end{align}  
denotes the downlink channel from the HAP to WD $k$. 

In phase 2, the MA at the HAP moves from the position $\sum_{m=1}^{N_\omega}s_1^m\mathbf p_m$ to the position $\sum_{m=1}^{N_\omega}s_2^m\mathbf p_m$, while the MA at WD $k$ moves from the position $\sum_{n=1}^{N_{u,k}}t_{k,1}^n\mathbf q_{k,n}$ to the position $\sum_{n=1}^{N_{u,k}}t_{k,2}^n\mathbf q_{k,n}$. Here, $s_2^m$ and $t_{k,2}^n$ are binary variables indicating the new target positions of the corresponding MAs in phase 2. We assume that the step time for each MA is $\Delta_d$ seconds. Then, the duration of this phase can be expressed as 
\begin{align}
	\tau_2 = & \max\left(\frac{\left\|\sum_{m=1}^{N_\omega}s_2^m\mathbf p_m -  \sum_{m=1}^{N_\omega}s_1^m\mathbf p_m\right\|_1}{d}\Delta_d\right., \nonumber\\
	&\left. \max_{k\in\mathcal K}\frac{\left\| \sum_{n=1}^{N_{u,k}}t_{k,2}^n\mathbf q_{k,n} - \sum_{n=1}^{N_{u,k}}t_{k,1}^n\mathbf q_{k,n}\right\|_1}{d}\Delta_d\right).
\end{align}

In the final phase, the system sum throughput can be expressed as  
\begin{align}
	R_{\rm sum}^{\rm disc} = \tau_3\log_2\left(1 + \sum_{k=1}^K\frac{p_k\left|h_{k,2}(\{s_2^m\},\{t_{k,2}^n\})\right|^2}{\sigma^2}\right), 
\end{align}
where $h_{k,2}(\{s_2^m\},\{t_{k,2}^n\})$ denotes the uplink channel from WD $k$ to the HAP, given by $h_{k,2}(\{s_2^m\},\{t_{k,2}^n\}) = \left( \left( \sum_{n=1}^{N_{u,k}}t_{k,2}^n\bm f_k(\mathbf q_{k,n})\right)^H\mathbf \Sigma_k\left( \sum_{m=1}^{N_\omega}s_2^m\bm g_k(\mathbf p_m)\right)\right)^H$. 

\subsection{Problem Formulation}
In this paper, our objective is to maximize the system sum throughput by jointly optimizing the MA positions, the time allocation, and the uplink power allocation. 
\subsubsection{\textbf{Continuous Antenna positioning}} The problem of interest can be mathematically formulated as
\begin{subequations}\label{P1}
	\begin{eqnarray}
		\hspace{-6mm}\text{(P1)}: \hspace{-1mm}&\underset{\mathcal Z_1}{\max}& \tau_3\log_2\left(1 + \sum_{k=1}^K\frac{p_k\left|h_{k,2}(\boldsymbol \omega_2, \boldsymbol u_{k,2})\right|^2}{\sigma^2}\right) \\
		&\text{s.t.}& \hspace{-3mm} \tau_1 + \tau_2 + \tau_3 \leq T, \label{P1_cons:b}\\
		&& \hspace{-3mm} 0 \leq \tau_1 \leq T, \ 0 \leq \tau_2 \leq T, \ 0 \leq \tau_3 \leq T, \label{P1_cons:c}\\
		&& \hspace{-3mm} p_k\tau_3 \leq \zeta P_{\rm A}\left|h_{k,1}(\boldsymbol \omega_1, \boldsymbol u_{k,1})\right|^2\tau_1, \ \forall k\in\mathcal K,\label{P1_cons:d}\\
		&& \hspace{-3mm} \boldsymbol \omega_1\in\mathcal C_\omega, \  \boldsymbol \omega_2\in\mathcal C_\omega, \label{P1_cons:e}\\
		&& \hspace{-3mm} \boldsymbol u_{k,1} \in\mathcal C_{u,k}, \ \boldsymbol u_{k,2} \in\mathcal C_{u,k}, \ \forall k\in\mathcal K, \label{P1_cons:f}
	\end{eqnarray}
\end{subequations}
where $\mathcal Z_1 \triangleq \left\lbrace \boldsymbol \omega_1,\boldsymbol \omega_2, \{\boldsymbol u_{k,1}\}, \{\boldsymbol u_{k,2}\},\tau_1, \tau_2, \tau_3, \{p_k \geq 0\} \right\rbrace$ is composed of all the optimization variables and \eqref{P1_cons:d} denotes the energy causality. Note that the optimization variables are intricately coupled in the objective function and constraint \eqref{P1_cons:d}. This renders (P1) to be a non-convex optimization problem that cannot be directly solved using standard optimization techniques.  

\subsubsection{\textbf{Discrete Antenna Positioning}} The corresponding sum throughput maximization problem can be formulated as
\begin{subequations}\label{P2}
	\begin{eqnarray}
	    &\hspace{-1cm}\text{(P2)}: &\hspace{-4mm} \underset{\mathcal Z_2}{\max} \hspace{1.5mm} \tau_3\log_2\left(1 \!+\! \sum_{k=1}^K\frac{p_k\left|h_{k,2}(\{s_2^m\},\{t_{k,2}^n\})\right|^2}{\sigma^2}\right) \\
	   &\hspace{-8mm} \text{s.t.}& \hspace{-3mm} \eqref{P1_cons:b},\eqref{P1_cons:c}, \label{P2_cons:b}\\
	   && \hspace{-3mm} p_k\tau_3 \leq \zeta P_{\rm A}\left|h_{k,1}(\{s_1^m\},\{t_{k,1}^n\})\right|^2\tau_1, \ \forall k\in\mathcal K, \label{P2_cons:c}\\
	   && \hspace{-3mm} s_1^m \in \{0,1\}, \ s_2^m \in \{0,1\}, \ \forall m \in\mathcal N_\omega, \label{P2_cons:d}\\
	   && \hspace{-3mm} \sum_{m=1}^{N_\omega}s_1^m = 1, \ \sum_{m=1}^{N_\omega}s_2^m = 1, \label{P2_cons:e}\\
	   && \hspace{-3mm} t_{k,1}^n \in \{0,1\}, \ t_{k,2}^n \in \{0,1\}, \ \forall k\in\mathcal K, n \in\mathcal N_{u,k}, \label{P2_cons:f}\\
	   && \hspace{-3mm} \sum_{n=1}^{N_{u,k}}t_{k,1}^n = 1, \ \sum_{n=1}^{N_{u,k}}t_{k,2}^n = 1, \ \forall k\in\mathcal K, \label{P2_cons:g}
	\end{eqnarray} 
\end{subequations}
where $\mathcal Z_2 \!\triangleq\!\left\lbrace \{s_1^m\}, \{s_2^m\}, \{t_{k,1}^n\}, \{t_{k,2}^n\}, \tau_1, \tau_2, \tau_3, \{p_k \!\geq\! 0\} \right\rbrace$. Problem (P2), being a mixed-integer non-convex optimization problem, is likely more difficult to solve to optimality than (P1) due to the presence of not only coupled variables but also binary variables, which generally increases the complexity of the solution process. 

\section{Proposed Solution for Continuous Positioning Design}\label{Sec:P1_solu}
In this section, we first investigate whether the optimal solution to (P1) requires distinct MA positions for the downlink WPT and uplink WIT phases. After that, we propose a computationally efficient algorithm to solve the resulting problem suboptimally. 

\vspace{-3mm}
\subsection{Should MA Positions Differ for downlink and uplink?}
For problem (P1), we have the following proposition.
\begin{prop}\label{prop1}
	\rm The optimal solution of (P1), denoted by $\mathcal Z_1^* \triangleq \left\lbrace \boldsymbol \omega_1^*,\boldsymbol \omega_2^*, \{\boldsymbol u_{k,1}^*\}, \{\boldsymbol u_{k,2}^*\}, \tau_1^*, \tau_2^*, \tau_3^*, \{p_k^* \geq 0\}\right\rbrace$, satisfies $p_k^* = \frac{\zeta P_{\rm A}\left|h_{k,1}(\boldsymbol \omega_1^*, \boldsymbol u_{k,1}^*)\right|^2\tau_1^*}{\tau_3^*}$, $\forall k\in\mathcal K$,
	$\boldsymbol \omega_1^* = \boldsymbol \omega_2^*$, $\boldsymbol u_{k,1}^* = \boldsymbol u_{k,2}^*$, $\forall k\in\mathcal K$, $\tau_2^* = 0$, and $\tau_3^* = T - \tau_1^*$.
\end{prop}
\begin{proof}
	Please refer to the appendix. 
\end{proof}

Proposition \ref{prop1} demonstrates that in the case of continuous positioning, utilizing identical MA positions for both downlink and uplink is the optimal strategy to maximize the system sum throughput. Moreover, as this strategy does not require a second MA movement, it is not only operationally efficient but also energy-efficient. By leveraging proposition \ref{prop1}, the design of MA positions is greatly simplified, reducing (P1) to the following formulation with much fewer variables:
\begin{subequations}\label{P1_simp}
	\begin{eqnarray}
		 &\hspace{-1cm} \underset{\substack{\tau_1, \boldsymbol \omega_1, \\ \{\boldsymbol u_{k,1}\}}}{\max}& \hspace{-4mm} \left(T \!-\! \tau_1\right) \log_2\left(1 + \sum_{k=1}^K\frac{\zeta P_{\rm A}\tau_1\left|h_{k,1}(\boldsymbol \omega_1, \boldsymbol u_{k,1})\right|^4}{\sigma^2\left(T -\tau_1\right)}\right) \label{P1_simp_obj}\\
		& \hspace{-1cm} \text{s.t.}& \hspace{-4mm} 0 \leq \tau_1 \leq T,  \label{P1_simp_cons:b}\\
		&& \hspace{-4mm} \boldsymbol \omega_1\in\mathcal C_\omega,  \label{P1_simp_cons:c}\\
		&& \hspace{-4mm} \boldsymbol u_{k,1} \in\mathcal C_{u,k}, \ \forall k\in\mathcal K. \label{P1_simp_cons:d}
	\end{eqnarray}
\end{subequations}
Despite this simplification, the problem remains non-convex due to the coupling of the optimization variables in the objective function. To tackle this, we employ the AO method to decouple these variables and iteratively update them, as detailed below. 

\subsection{Proposed Algorithm for Problem \eqref{P1_simp}}\label{Sec:P1_simp_solu}
\subsubsection{Optimizing $\boldsymbol \omega_1$}
For any given $\left\lbrace\tau_1,\{\boldsymbol u_{k,1}\}\right\rbrace$,  $\boldsymbol \omega_1$ can be optimized by maximizing the expression inside the logarithm in \eqref{P1_simp_obj}, as follows
\begin{eqnarray}\label{P1_simp_sub1}
		\underset{\boldsymbol \omega_1}{\max} \hspace{2mm} \sum_{k=1}^K\mu_k\left|h_{k,1}(\boldsymbol \omega_1, \boldsymbol u_{k,1})\right|^4 \hspace{8mm}
		\text{s.t.} \hspace{2mm} \eqref{P1_simp_cons:c},
\end{eqnarray}
where $\mu_k \triangleq \frac{\zeta P_{\rm A}\tau_1}{\sigma^2\left(T- \tau_1\right)}$, $\forall k\in\mathcal K$. Note that the optimization variable $\boldsymbol \omega_1$ does not appear explicitly in the current form of the objective function. Recall that $h_{k,1}(\boldsymbol \omega_1, \boldsymbol u_{k,1}) = \bm f_k(\boldsymbol u_{k,1})^H\mathbf \Sigma_k\bm g_k(\boldsymbol \omega_1)$. To facilitate the solution of problem \eqref{P1_simp_sub1}, we define $\bm b_k^H \triangleq \bm f_k(\boldsymbol u_{k,1})^H\mathbf \Sigma_k \in \mathbb C^{1\times L_k}$ and $\bm B_k \triangleq \bm b_k\bm b_k^H \in \mathbb C^{L_k\times L_k}$, $\forall k\in\mathcal K$. With these definitions, we expand the term $\left|h_{k,1}(\boldsymbol \omega_1, \boldsymbol u_{k,1})\right|^4$ as
\begin{align}\label{eq:omega_expan}
	& \left|h_{k,1}(\boldsymbol \omega_1, \boldsymbol u_{k,1})\right|^4  = \left|\bm b_k^H\bm g_k(\boldsymbol \omega_1)\right|^4  = \left(\bm g_k^H(\boldsymbol \omega_1)\bm B_k\bm g_k(\boldsymbol \omega_1)\right)^2 \nonumber\\   
	& = \sum_{i_1=1}^{L_k}\sum_{i_2=1}^{L_k}\sum_{i_3=1}^{L_k}\sum_{i_4=1}^{L_k}\left|\left[\bm B_k\right]_{i_1,i_2}\right|\left|\left[\bm B_k\right]_{i_3,i_4}\right| \nonumber\\
	& \times e^{\jmath\left( \frac{2\pi}{\lambda}\boldsymbol \omega_1^T\left(-\boldsymbol a_k^{i_1} + \boldsymbol a_k^{i_2} -\boldsymbol a_k^{i_3} + \boldsymbol a_k^{i_4} \right) + \arg\left(\left[\boldsymbol B_k\right]_{i_1,i_2}\right) + \arg\left(\left[\boldsymbol B_k\right]_{i_3,i_4}\right)\right)} \nonumber\\
	& = \sum_{i_1=1}^{L_k}\sum_{i_2=1}^{L_k}\sum_{i_3=1}^{L_k}\sum_{i_4=1}^{L_k}\left|\left[\bm B_k\right]_{i_1,i_2}\right|\left|\left[\bm B_k\right]_{i_3,i_4}\right|\nonumber\\
	& \times \cos\Bigg( \frac{2\pi}{\lambda}\boldsymbol \omega_1^T\left(-\boldsymbol a_k^{i_1} + \boldsymbol a_k^{i_2} -\boldsymbol a_k^{i_3} + \boldsymbol a_k^{i_4} \right) + \arg\left(\left[\boldsymbol B_k\right]_{i_1,i_2}\right) \nonumber\\ 
	& + \arg\left(\left[\boldsymbol B_k\right]_{i_3,i_4}\right)\Bigg)  \nonumber\\
	& \triangleq \sum_{i_1=1}^{L_k}\sum_{i_2=1}^{L_k}\sum_{i_3=1}^{L_k}\sum_{i_4=1}^{L_k}\left|\left[\bm B_k\right]_{i_1,i_2}\right|\left|\left[\bm B_k\right]_{i_3,i_4}\right| \nonumber\\
	& \times \cos\left(\varpi_{k,i_1,i_2,i_3,i_4}(\boldsymbol \omega_1)\right) \triangleq \Omega_k(\boldsymbol \omega_1),  
\end{align}
where $\varpi_{k,i_1,i_2,i_3,i_4}(\boldsymbol \omega_1) \triangleq \frac{2\pi}{\lambda}\boldsymbol \omega_1^T\left(-\boldsymbol a_k^{i_1} + \boldsymbol a_k^{i_2} -\boldsymbol a_k^{i_3} + \boldsymbol a_k^{i_4} \right) + \arg\left(\left[\boldsymbol B_k\right]_{i_1,i_2}\right) + \arg\left(\left[\boldsymbol B_k\right]_{i_3,i_4}\right)$. Observe that $\Omega_k(\boldsymbol \omega_1)$ does not exhibit concavity or convexity with respect to $\boldsymbol \omega_1$, rendering the maximization of $\sum_{k=1}^K\mu_k\Omega_k(\boldsymbol \omega_1)$ a non-convex problem. Nevertheless, we find that $\Omega_k(\boldsymbol \omega_1)$ possesses bounded curvature, i.e., there exists a positive real number $\psi_k$ such that $\psi_k\mathbf I \succeq \nabla^2 \Omega_k(\boldsymbol \omega_1)$. This allows us to leverage the SCA technique to solve this problem. Specifically, with given local point $\boldsymbol \omega_1^r$ in the $r$-th iteration, we can obtain the following global lower bound for $\Omega_k(\boldsymbol \omega_1)$ by modifying \cite[(25)]{2017_Sun_MM}: 
\begin{align}\label{ineq:surrogate_func_omega}
	\Omega_k(\boldsymbol \omega_1) & \geq  \Omega_k(\boldsymbol \omega_1^r) + \nabla \Omega_k(\boldsymbol \omega_1^r)^T(\boldsymbol \omega_1 - \boldsymbol \omega_1^r) \nonumber\\
	& \hspace{4mm} - \frac{\psi_k}{2}(\boldsymbol \omega_1 - \boldsymbol \omega_1^r)^T(\boldsymbol \omega_1 - \boldsymbol \omega_1^r) \triangleq \Omega_k^{\rm lb,\it r}(\boldsymbol \omega_1), 
\end{align} 
where $\nabla \Omega_k(\boldsymbol \omega_1^r)^T = \left[\frac{\partial \Omega_k(\boldsymbol \omega_1)}{\partial x_{\omega,1}}\Big|_{\boldsymbol \omega_1 =\boldsymbol \omega_1^r}, \frac{\partial \Omega_k(\boldsymbol \omega_1)}{\partial y_{\omega,1}}\Big|_{\boldsymbol \omega_1 =\boldsymbol \omega_1^r}\right]$ with
\begin{subequations}\label{eq:first-order}
	\begin{align}
		& \frac{\partial \Omega_k(\boldsymbol \omega_1)}{\partial x_{\omega,1}}\Big|_{\boldsymbol \omega_1 = \boldsymbol \omega_1^r} =  -\frac{2\pi}{\lambda}\sum_{i_1=1}^{L_k}\sum_{i_2=1}^{L_k}\sum_{i_3=1}^{L_k}\sum_{i_4=1}^{L_k}\left|\left[\bm B_k\right]_{i_1,i_2}\right|\nonumber\\
		& \hspace{5mm}\times \left|\left[\bm B_k\right]_{i_3,i_4}\right|\alpha_{k,i_1,i_2,i_3,i_4}\sin\left(\varpi_{k,i_1,i_2,i_3,i_4}(\boldsymbol \omega_1^r)\right), \\
		& \frac{\partial \Omega_k(\boldsymbol \omega_1)}{\partial y_{\omega,1}}\Big|_{\boldsymbol \omega_1 = \boldsymbol \omega_1^r} =  -\frac{2\pi}{\lambda}\sum_{i_1=1}^{L_k}\sum_{i_2=1}^{L_k}\sum_{i_3=1}^{L_k}\sum_{i_4=1}^{L_k}\left|\left[\bm B_k\right]_{i_1,i_2}\right|\nonumber\\
		& \hspace{5mm} \times \left|\left[\bm B_k\right]_{i_3,i_4}\right|\beta_{k,i_1,i_2,i_3,i_4}\sin\left(\varpi_{k,i_1,i_2,i_3,i_4}(\boldsymbol \omega_1^r)\right).
	\end{align}
\end{subequations}
In \eqref{eq:first-order}, $\alpha_{k,i_1,i_2,i_3,i_4} \triangleq -\sin\theta_k^{i_1}\cos\phi_k^{i_1} + \sin\theta_k^{i_2}\cos\phi_k^{i_2} - \sin\theta_k^{i_3}\cos\phi_k^{i_3} + \sin\theta_k^{i_4}\cos\phi_k^{i_4}$ and $\beta_{k,i_1,i_2,i_3,i_4} \triangleq -\cos\theta_k^{i_1} + \cos\theta_k^{i_2} - \cos\theta_k^{i_3} + \cos\theta_k^{i_4}$. Moreover, the value of $\psi_k$ that satisfies $\psi_k\mathbf I \succeq \nabla^2 \Omega_k(\boldsymbol \omega_1)$ can be determined by choosing $\psi_k$ such that $\psi_k \geq \left\|\nabla^2 \Omega_k(\boldsymbol \omega_1)\right\|_F$, since $\left\|\nabla^2 \Omega_k(\boldsymbol \omega_1)\right\|_F\mathbf I \succeq \left\|\nabla^2 \Omega_k(\boldsymbol \omega_1)\right\|_2\mathbf I \succeq \nabla^2 \Omega_k(\boldsymbol \omega_1)$. To proceed,  $\left\|\nabla^2 \Omega_k(\boldsymbol \omega_1)\right\|_F$ can be computed as 
\begin{align}
	& \left\|\nabla^2 \Omega_k(\boldsymbol \omega_1)\right\|_F \nonumber\\
	& = \Bigg[\left(\frac{\partial \Omega_k(\boldsymbol \omega_1)}{\partial x_{\omega,1}\partial x_{\omega,1}}\right)^2 + \left(\frac{\partial \Omega_k(\boldsymbol \omega_1)}{\partial x_{\omega,1}\partial y_{\omega,1}}\right)^2 + \left(\frac{\partial \Omega_k(\boldsymbol \omega_1)}{\partial y_{\omega,1}\partial x_{\omega,1}}\right)^2 \nonumber\\
	& \hspace{4mm} + \left(\frac{\partial \Omega_k(\boldsymbol \omega_1)}{\partial y_{\omega,1}\partial y_{\omega,1}}\right)^2\Bigg]^{\frac{1}{2}},
\end{align}
where 
\begin{subequations}\label{eq:second-order}
	\begin{align}
		& \frac{\partial \Omega_k(\boldsymbol \omega_1)}{\partial x_{\omega,1}\partial x_{\omega,1}} = -\frac{4\pi^2}{\lambda^2}\sum_{i_1=1}^{L_k}\sum_{i_2=1}^{L_k}\sum_{i_3=1}^{L_k}\sum_{i_4=1}^{L_k}\left|\left[\bm B_k\right]_{i_1,i_2}\right|\nonumber\\
		& \hspace{5mm}\times \left|\left[\bm B_k\right]_{i_3,i_4}\right|\alpha_{k,i_1,i_2,i_3,i_4}^2\cos\left(\varpi_{k,i_1,i_2,i_3,i_4}(\boldsymbol \omega_1)\right), \\
		& \frac{\partial \Omega_k(\boldsymbol \omega_1)}{\partial x_{\omega,1}\partial y_{\omega,1}} = \frac{\partial \Omega_k(\boldsymbol \omega_1)}{\partial y_{\omega,1}\partial x_{\omega,1}} \nonumber\\
		& =   -\frac{4\pi^2}{\lambda^2}\sum_{i_1=1}^{L_k}\sum_{i_2=1}^{L_k}\sum_{i_3=1}^{L_k}\sum_{i_4=1}^{L_k}\left|\left[\bm B_k\right]_{i_1,i_2}\right|\left|\left[\bm B_k\right]_{i_3,i_4}\right|\nonumber\\
		& \hspace{5mm}\times\alpha_{k,i_1,i_2,i_3,i_4}\beta_{k,i_1,i_2,i_3,i_4}\cos\left(\varpi_{k,i_1,i_2,i_3,i_4}(\boldsymbol \omega_1)\right), \\
		& \frac{\partial \Omega_k(\boldsymbol \omega_1)}{\partial y_{\omega,1}\partial y_{\omega,1}} =  -\frac{4\pi^2}{\lambda^2}\sum_{i_1=1}^{L_k}\sum_{i_2=1}^{L_k}\sum_{i_3=1}^{L_k}\sum_{i_4=1}^{L_k}\left|\left[\bm B_k\right]_{i_1,i_2}\right|\nonumber\\
		& \hspace{5mm}\times \left|\left[\bm B_k\right]_{i_3,i_4}\right|\beta_{k,i_1,i_2,i_3,i_4}^2\cos\left(\varpi_{k,i_1,i_2,i_3,i_4}(\boldsymbol \omega_1)\right). 
	\end{align}
\end{subequations}
Assuming that $\cos\left(\varpi_{k,i_1,i_2,i_3,i_4}(\boldsymbol \omega_1)\right) = 1$, an upper bound for $\left\|\nabla^2 \Omega_k(\boldsymbol \omega_1)\right\|_F$ can be obtained, and this bound can be selected as the value of $\psi_k$. 

With \eqref{eq:omega_expan} and \eqref{ineq:surrogate_func_omega}, the objective function of problem \eqref{P1_simp_sub1} is lower bounded by $\sum_{k=1}^K\mu_k\Omega_k^{\rm lb,\it r}(\boldsymbol \omega_1)$, which is a concave function. As a result, a lower bound of the optimal value of problem \eqref{P1_simp_sub1} can be obtained by solving the following convex problem with off-the-shelf solvers (such as CVX \cite{2004_S.Boyd_cvx}). 
\begin{eqnarray}\label{P1_simp_sub1_sca}
	\underset{\boldsymbol \omega_1}{\max} \hspace{2mm} \sum_{k=1}^K\mu_k\Omega_k^{\rm lb,\it \ell}(\boldsymbol \omega_1) \hspace{8mm}
	\text{s.t.} \hspace{2mm} \eqref{P1_simp_cons:c}.
\end{eqnarray}

\subsubsection{Optimizing $\{\boldsymbol u_{k,1}\}$}
For any given $\left\lbrace\boldsymbol \omega_1, \tau_1\right\rbrace $, $\{\boldsymbol u_{k,1}\}$ can be optimized by solving (P1) with only constraint \eqref{P1_simp_cons:d}. Since $\{\boldsymbol u_{k,1}\}$ are separable in both the objective function and the constraint, the problem can be decomposed into $K$ independent subproblems, each corresponding to a different $k\in\mathcal K$:  
\begin{subequations}\label{P1_simp_sub2}
	\begin{eqnarray}
		&\underset{\boldsymbol u_{k,1}}{\max}& \left|h_{k,1}(\boldsymbol \omega_1, \boldsymbol u_{k,1})\right|^4\\
		&\text{s.t.}& \hspace{-1mm} \boldsymbol u_{k,1} \in\mathcal C_{u,k}, \ k\in\mathcal K, \label{P1_simp_sub2_cons:b}
	\end{eqnarray}
\end{subequations}
where the logarithm and certain constant terms in the original objective function are omitted without affecting the optimality of $\boldsymbol u_{k,1}$.  Similar to \eqref{eq:omega_expan}, we expand $\left|h_{k,1}(\boldsymbol \omega_1, \boldsymbol u_{k,1})\right|^4$ to expose $\boldsymbol u_{k,1}$ as follows:
\begin{align}\label{eq:u_expan}
	& \left|h_{k,1}(\boldsymbol \omega_1, \boldsymbol u_{k,1})\right|^4 = \left|\bm f_k(\boldsymbol u_{k,1})^H\mathbf \Sigma_k\bm g_k(\boldsymbol \omega_1)\right|^4 \nonumber\\
	& = (\bm f_k(\boldsymbol u_{k,1})^H\bm D_k\bm f_k(\boldsymbol u_{k,1}))^2 \nonumber\\
	& = \sum_{j_1=1}^{\tilde L_k}\sum_{j_2=1}^{\tilde L_k}\sum_{j_3=1}^{\tilde L_k}\sum_{j_4=1}^{\tilde L_k}\left|\left[\bm D_k\right]_{j_1,j_2}\right|\left|\left[\bm D_k\right]_{j_3,j_4}\right| \nonumber\\
	& \times e^{\jmath\left(\frac{2\pi}{\lambda}\boldsymbol u_{k,1}^T\left(-\tilde{\boldsymbol a}_k^{j_1} + \tilde{\boldsymbol a}_k^{j_2} - \tilde{\boldsymbol a}_k^{j_3} + \tilde{\boldsymbol a}_k^{j_4} \right) + \arg\left(\left[\boldsymbol D_k\right]_{j_1,j_2}\right) + \arg\left(\left[\boldsymbol D_k\right]_{j_3,j_4}\right)\right)} \nonumber\\
	& = \sum_{j_1=1}^{\tilde L_k}\sum_{j_2=1}^{\tilde L_k}\sum_{j_3=1}^{\tilde L_k}\sum_{j_4=1}^{\tilde L_k}\left|\left[\bm D_k\right]_{j_1,j_2}\right|\left|\left[\bm D_k\right]_{j_3,j_4}\right| \nonumber\\
	& \times \cos\Bigg(\frac{2\pi}{\lambda}\boldsymbol u_{k,1}^T\left(-\tilde{\boldsymbol a}_k^{j_1} + \tilde{\boldsymbol a}_k^{j_2} - \tilde{\boldsymbol a}_k^{j_3} \!+\! \tilde{\boldsymbol a}_k^{j_4} \right) \!+\! \arg\left(\left[\boldsymbol D_k\right]_{j_1,j_2}\right) \nonumber\\ 
	& \hspace{2mm} + \arg\left(\left[\boldsymbol D_k\right]_{j_3,j_4}\right)\Bigg) \triangleq U(\boldsymbol u_{k,1}), 
\end{align}
where $\bm D_k \triangleq \mathbf \Sigma_k\bm g_k(\boldsymbol \omega_1)\bm g_k(\boldsymbol \omega_1)^H\mathbf \Sigma_k^H$. Since $U(\boldsymbol u_{k,1})$ has a similar form as $\Omega_k(\boldsymbol \omega_1)$ in \eqref{eq:omega_expan}, it can be handled similarly as for $\Omega_k(\boldsymbol \omega_1)$. To be specific, by applying the second-order Taylor expansion, we can obtain a global lower bound for $\boldsymbol u_{k,1}$, denoted by $U^{\rm lb,\it r}(\boldsymbol u_{k,1}) \triangleq U(\boldsymbol u_{k,1}^r) + \nabla U(\boldsymbol u_{k,1}^r)^T(\boldsymbol u_{k,1} - \boldsymbol u_{k,1}^r) - \frac{\delta_k}{2}(\boldsymbol u_{k,1} - \boldsymbol u_{k,1}^r)^T(\boldsymbol u_{k,1} - \boldsymbol u_{k,1}^r)$, where $\boldsymbol u_{k,1}^r$ is the given local point in the $r$-th iteration and $\delta_k$ is a positive real number satisfying $\delta_k\mathbf I \succeq \nabla^2 U(\boldsymbol u_{k,1})$. Then, problem \eqref{P1_simp_sub2} can be approximated as (with constant terms dropped) 
\begin{subequations}\label{P1_simp_sub2_sca}
	\begin{eqnarray}
		&\hspace{-7mm}\underset{\boldsymbol u_{k,1}}{\max}& \hspace{-1mm} - \frac{\delta_k}{2}\boldsymbol u_{k,1}^T\boldsymbol u_{k,1} + \left(\nabla U(\boldsymbol u_{k,1}^r) + \delta_k\boldsymbol u_{k,1}^r\right)^T\boldsymbol u_{k,1}\\
		&\hspace{-7mm} \text{s.t.}& \hspace{-2mm} \eqref{P1_simp_sub2_cons:b}.
	\end{eqnarray}
\end{subequations} 
If constraint \eqref{P1_simp_sub2_cons:b} is disregarded, the objective function can be maximized with a closed-form solution given by $\boldsymbol u_{k,1}^{\star} = \frac{\nabla U(\boldsymbol u_{k,1}^r)}{\delta_k} + \boldsymbol u_{k,1}^r$. If this solution meets constraint \eqref{P1_simp_sub2_cons:b}, it constitutes the optimal solution to problem \eqref{P1_simp_sub2_sca}. Otherwise, one can solve the convex problem \eqref{P1_simp_sub2_sca} optimally using standard solvers, e.g., CVX \cite{2004_S.Boyd_cvx}.

\subsubsection{Optimizing $\tau_1$}\label{Sec:P1_simp_sub1}
When all the MA positions are given, (P1) reduces to
\begin{eqnarray}\label{P1_simp_sub3}
	\underset{\tau_1}{\max} \hspace{2mm} \left(T-\tau_1\right)\log_2\left(1 + \frac{c\tau_1}{T-\tau_1}\right) \hspace{8mm}
	\text{s.t.} \hspace{2mm} \eqref{P1_simp_cons:b}, 
\end{eqnarray}
where $c\triangleq \frac{\sum_{k=1}^K{\zeta_kP_{\rm A}\left|h_{k,1}(\boldsymbol \omega_1, \boldsymbol u_{k,1})\right|^4}}{\sigma^2}$. The optimal solution of this problem is given by \cite[Theorem 1]{2022_Zheng_timesolution}
\begin{align}\label{eq:P1_simp_sub3_solution}
	\tau_1^{\star} = \frac{T\left(\exp\left(\mathcal W\left(\frac{c-1}{e}\right) + 1\right)-1\right) }{c+\exp\left(\mathcal W\left(\frac{c-1}{e}\right) + 1\right)-1},
\end{align}
where $\mathcal W\left(\cdot\right)$ represents the Lambert function.  

\begin{algorithm}[!t]  
	\caption{Proposed AO-based algorithm for problem \eqref{P1_simp}}  \label{Alg1}  
	\begin{algorithmic}[1]
		\STATE Initialize $\left\lbrace \boldsymbol \omega_1, \left\lbrace \boldsymbol u_{k,1}\right\rbrace, \tau_1\right\rbrace$ and set $r = 0$.  
		\REPEAT 
		\STATE Compute $\{\bm B_k\}$, $\{\nabla \Omega_k(\boldsymbol \omega_1^r)\}$, $\{\nabla^2 \Omega_k(\boldsymbol \omega_1)\}$, and $\{\psi_k\}$. 
		\vspace{-3.5mm}
		\STATE Update $\boldsymbol \omega_1^{r+1}$ by solving problem \eqref{P1_simp_sub1_sca} with given $\big\{\boldsymbol \omega_1^r, \{\boldsymbol u_{k,1}^r\}, \tau_1^r\big\}$.  \label{Alg:solve_sub1}
		\STATE Compute $\{\bm D_k\}$, $\{\nabla U(\boldsymbol u_{k,1}^r)\}$, $\{\nabla^2 U(\boldsymbol u_{k,1}^r)\}$, and $\{\delta_k\}$. 
		\STATE Update $\{\boldsymbol u_{k,1}^r\}$ by solving problem \eqref{P1_simp_sub2_sca} with given $\big\{\boldsymbol \omega_1^{r+1}, \{\boldsymbol u_{k,1}^r\}, \tau_1^r\big\}$. \label{Alg:solve_sub2}
		\STATE Update $\tau_1^{r+1}$ via \eqref{eq:P1_simp_sub3_solution} with given $\big\{\boldsymbol \omega_1^{r+1}, \{ \boldsymbol u_{k,1}^{r+1}\}, \tau_1^r\big\}$.  \label{Alg:solve_sub3}  
		\vspace{-3.5mm}
		\STATE $r \leftarrow r + 1$. 
		\UNTIL The fractional increase of the objective value drops below a threshold $\epsilon > 0$. 
	\end{algorithmic} 
\end{algorithm}

\subsubsection{Convergence and Complexity Analysis}
In summary, problem \eqref{P1_simp} is solved suboptimally by alternately updating the three subsets of variables, with details summarized in Algorithm \ref{Alg1}. In the following, we provide a detailed convergence proof of the proposed algorithm. Denote by $\eta\left(\boldsymbol \omega_1, \left\lbrace \boldsymbol u_{k,1} \right\rbrace, \tau_1\right)$, $\eta_1^{\rm lb}\left(\boldsymbol \omega_1, \left\lbrace \boldsymbol u_{k,1} \right\rbrace, \tau_1\right)$, and $\eta_2^{\rm lb}\left(\boldsymbol \omega_1, \left\lbrace \boldsymbol u_{k,1} \right\rbrace, \tau_1\right)$ the objective values of problems \eqref{P1_simp}, \eqref{P1_simp_sub1_sca}, and \eqref{P1_simp_sub2_sca}, respectively, evaluated at the point $\left(\boldsymbol \omega_1, \left\lbrace \boldsymbol u_{k,1} \right\rbrace, \tau_1\right)$. In step \ref{Alg:solve_sub1} of Algorithm \ref{Alg1}, we have
\begin{align}\label{Alg1_proof_1}
	\eta\left(\boldsymbol \omega_1^r, \left\lbrace \boldsymbol u_{k,1}^r\right\rbrace, \tau_1^r\right) & \overset{(a)}{=} \eta_1^{\rm lb}\left(\boldsymbol \omega_1^r, \left\lbrace \boldsymbol u_{k,1}^r\right\rbrace, \tau_1^r\right) \nonumber\\
	& \overset{(b)}{\leq} \eta_1^{\rm lb}\left(\boldsymbol \omega_1^{r+1}, \left\lbrace \boldsymbol u_{k,1}^r\right\rbrace, \tau_1^r\right) \nonumber\\
	& \overset{(c)}{\leq} \eta\left(\boldsymbol \omega_1^{r+1}, \left\lbrace \boldsymbol u_{k,1}^r\right\rbrace, \tau_1^r\right),
\end{align}
where $(a)$ holds since the second-order Taylor expansion in \eqref{ineq:surrogate_func_omega} is tight at $\omega_1^r$; $(b)$ follows from the update of $\boldsymbol \omega_1$ from $\boldsymbol \omega_1^r$ to $\boldsymbol \omega_1^{r+1}$, which maximizes the surrogate objective value $\eta_1^{\rm lb}\left(\boldsymbol \omega_1^{r+1}, \left\lbrace \boldsymbol u_{k,1}^r\right\rbrace, \tau_1^r\right)$ in step \ref{Alg:solve_sub1}; $(c)$ is due to the fact that $\eta_1^{\rm lb}$ serves as a valid lower bound to the original objective $\eta$. Similarly, in step \ref{Alg:solve_sub2}, we obtain the following: 
\begin{align}\label{Alg1_proof_2}
	\eta\left(\boldsymbol \omega_1^{r+1}, \left\lbrace \boldsymbol u_{k,1}^r\right\rbrace, \tau_1^r\right) & = \eta_2^{\rm lb}\left(\boldsymbol \omega_1^{r+1}, \left\lbrace \boldsymbol u_{k,1}^r\right\rbrace, \tau_1^r\right) \nonumber\\
	& \leq \eta_2^{\rm lb}\left(\boldsymbol \omega_1^{r+1}, \left\lbrace \boldsymbol u_{k,1}^{r+1}\right\rbrace, \tau_1^r\right) \nonumber\\
	& \leq \eta\left(\boldsymbol \omega_1^{r+1}, \left\lbrace \boldsymbol u_{k,1}^{r+1}\right\rbrace, \tau_1^r\right), 
\end{align}
where the equality and inequalities follow analogously to \eqref{Alg1_proof_1}. 
Next, in step \ref{Alg:solve_sub3}, the variable $\tau_1$ is updated by solving problem \eqref{P1_simp_sub3}, which is a simplified subproblem of problem \eqref{P1_simp} with fixed $\boldsymbol \omega_1^{r+1}$ and $\left\lbrace \boldsymbol u_{k,1}^{r+1} \right\rbrace$. As this subproblem is solved optimally, we have 
\begin{align}\label{Alg1_proof_3}
	\eta\left(\boldsymbol \omega_1^{r+1}, \left\lbrace \boldsymbol u_{k,1}^{r+1}\right\rbrace, \tau_1^r\right) \leq \eta\left(\boldsymbol \omega_1^{r+1}, \left\lbrace \boldsymbol u_{k,1}^{r+1}\right\rbrace, \tau_1^{r+1}\right).
\end{align}
Combining the results in \eqref{Alg1_proof_1}-\eqref{Alg1_proof_3}, we establish the monotonicity of the objective: 
\begin{align}
	\eta\left(\boldsymbol \omega_1^r, \left\lbrace \boldsymbol u_{k,1}^r\right\rbrace, \tau_1^r\right) \leq \eta\left(\boldsymbol \omega_1^{r+1}, \left\lbrace \boldsymbol u_{k,1}^{r+1}\right\rbrace, \tau_1^{r+1}\right),
\end{align}
i.e., the objective value of problem \eqref{P1_simp} is non-decreasing over iterations. Furthermore, as $\eta$ is upper bounded due to physical and design constraints, the convergence of Algorithm \ref{Alg1} is guaranteed.

We now turn to analyzing the complexity of the proposed algorithm. Clearly, the primary computational cost per iteration is due to the steps involved in updating $\boldsymbol \omega_1$ and $\{\boldsymbol u_{k,1}\}$. For updating $\boldsymbol \omega_1$, the complexities of computing $\{\bm B_k\}$, $\left\lbrace \nabla \Omega_k(\boldsymbol \omega_1^r)\right\rbrace$, $\left\lbrace \nabla^2 \Omega_k(\boldsymbol \omega_1)\right\rbrace$, and $\{\psi_k\}$, as well as solving problem \eqref{P1_simp_sub1_sca}, are $\mathcal O\left(\sum_{k=1}^K\left( \tilde L_kL_k + L_k^2\right) \right)$, $\mathcal O\left(\sum_{k=1}^KL_k^4\right)$, $\mathcal O\left(\sum_{k=1}^KL_k^4\right)$, $\mathcal O\left(K\right)$, and $\mathcal O(1)$, respectively. Similarly, for updating $\{\boldsymbol u_{k,1}\}$, the complexities of computing $\{\bm D_k\}$, $\left\lbrace \nabla U(\boldsymbol u_{k,1}^r)\right\rbrace$, $\left\lbrace \nabla^2 U(\boldsymbol u_{k,1}^r)\right\rbrace$, and $\{\delta_k\}$, and solving $K$ subproblems of the form \eqref{P1_simp_sub2_sca}, are $\mathcal O\left(\sum_{k=1}^K\left( \tilde L_kL_k + \tilde L_k^2\right) \right)$, $\mathcal O\left(\sum_{k=1}^K\tilde L_k^4\right)$, $\mathcal O\left(\sum_{k=1}^K\tilde L_k^4\right)$, $\mathcal O\left(K\right)$, and $\mathcal O(K)$, respectively. Combining these complexities, the total complexity of each iteration of the proposed algorithm is about $\mathcal O\left(\sum_{k=1}^K\left(\tilde L_kL_k + \tilde L_k^4 + L_k^4\right)\right)$. 

\section{Proposed Solution for Discrete Positioning Design}\label{Sec:P2_solu}
In this section, our focus shifts to solving (P2). As in the previous section, we begin by exploring whether achieving the optimum of (P2) needs different MA positions in downlink and uplink. Then, we develop an iterative algorithm to efficiently address the resulting problem.

\vspace{-2mm}
\subsection{Should MA Positions Differ for downlink and uplink?}
For problem (P2), we have the following proposition.
\begin{prop}\label{prop2}
	\rm The optimal solution of (P2), denoted by $\mathcal Z_2^* \triangleq \left\lbrace \{s_1^{m*}\}, \{s_2^{m*}\}, \{t_{k,1}^{n*}\}, \{t_{k,2}^{n *}\}, \tau_1^*, \tau_2^*, \tau_3^*, \{p_k^* \geq 0\}\right\rbrace$, satisfies $p_k^* = \frac{\zeta P_{\rm A}\left|h_{k,1}(\{s_1^{m*}\},\{t_{k,1}^{n*}\})\right|^2\tau_1^*}{\tau_3^*}$, $\forall k\in\mathcal K$, 
	$s_1^{m*} = s_2^{m*}$, $t_{k,1}^{n*} = t_{k,2}^{n*}$, $\forall k\in\mathcal K$, $\tau_2^* = 0$, and $\tau_3^* = T - \tau_1$.  
\end{prop}
\begin{proof}
	The proof is similar to that of Proposition \ref{prop1} given in the appendix, and we omit it for brevity. 
\end{proof}
We observe that the results of Proposition \ref{prop2} are analogous to those in Proposition \ref{prop1}, albeit for a different MA movement pattern. This indicates that regardless of whether the movement pattern is continuous or discrete, using identical MA positions for both downlink and uplink is the optimal strategy to maximize the system sum throughput. Essentially, the key to this optimality lies in maintaining consistent channel conditions for both the downlink and uplink transmissions. According to Proposition \ref{prop2}, we only need to focus on solving the following simplified problem: 
\begin{subequations}\label{P2_simp}
	\begin{eqnarray}
		&\hspace{-2mm} \underset{\substack{\tau_1, \{s_1^m\},\\ \{t_{k,1}^n\}}}{\max}& \hspace{-4mm} \left(T \!-\! \tau_1\right) \log_2\left(1 \!+\! \sum_{k=1}^K\frac{\zeta P_{\rm A}\tau_1\left|h_{k,1}(\{s_1^m\},\{t_{k,1}^n\})\right|^4}{\sigma^2\left(T-\tau_1\right) }\right) \nonumber\\
		\\
		&\text{s.t.}& \hspace{-3mm}  0 \leq \tau_1 \leq T, \label{P2_simp_cons:b}\\
		&& \hspace{-3mm} s_1^m \in \{0,1\}, \ \forall m \in\mathcal N_\omega, \label{P2_simp_cons:c}\\
		&& \hspace{-3mm} \sum_{m=1}^{N_\omega}s_1^m = 1, \  \label{P2_simp_cons:d}\\
		&& \hspace{-3mm} t_{k,1}^n \in \{0,1\}, \ \forall k\in\mathcal K, n \in\mathcal N_{u,k},\label{P2_simp_cons:e}\\
		&& \hspace{-3mm} \sum_{n=1}^{N_{u,k}}t_{k,1}^n = 1, \ \forall k\in\mathcal K, \label{P2_simp_cons:f} 
	\end{eqnarray}
\end{subequations}
which is still a mixed-integer non-convex optimization problem. To solve it, we apply the AO method to divide the optimization variables into three blocks, as elaborated below. 

\subsection{Proposed Algorithm for Problem \eqref{P2_simp}}\label{Sec:P2_simp_solu}
\subsubsection{Optimizing $\{s_1^m\}$} 
For given $\left\lbrace\tau_1, \{t_{k,1}^n\}\right\rbrace$,  $\{s_1^m\}$ can be optimized by solving 
\begin{subequations}\label{P2_simp_sub1}
	\begin{eqnarray}
		&\hspace{-4mm}\underset{\{s_1^m\}}{\max}& \hspace{-3.5mm} \sum_{k=1}^K\mu_k\left|\left( \sum_{n=1}^{N_{u,k}}t_{k,1}^n\bm f_k(\mathbf q_{k,n})\right)^H \! \mathbf \Sigma_k \!\left(\sum_{m=1}^{N_\omega}s_1^m\bm g_k(\mathbf p_m)\right)\right|^4 \nonumber\\
		\\
		&\text{s.t.}& \hspace{-2mm} \eqref{P2_simp_cons:c}, \eqref{P2_simp_cons:d},
	\end{eqnarray}
\end{subequations}
where $\mu_k$ is defined as $\frac{\zeta P_{\rm A}\tau_1}{\sigma^2\left(T- \tau_1\right)}$, as introduced in the previous section. To facilitate the solution design, we further define $\mathbf y_k^H \triangleq \left(\sum_{n=1}^{N_{u,k}}t_{k,1}^n\bm f_k(\mathbf q_{k,n})\right)^H\mathbf \Sigma_k \in \mathbb C^{1\times L_k}$, 
$\mathbf G_k \triangleq \left[\bm g_k(\mathbf p_1), \bm g_k(\mathbf p_2), \cdots, \bm g_k(\mathbf p_{L_k})\right] \in \mathbb C^{L_k \times N_\omega}$, $\mathbf o_k^H \triangleq \mathbf y_k^H\mathbf G_k \in \mathbb C^{1 \times N_\omega}$, and $\bm s_1 \triangleq \left[s_1^1,s_1^2,\cdots, s_1^{N_\omega}\right]^T \in \mathbb R^{N_\omega\times 1}$. Then, we have $\left|\left(\sum_{n=1}^{N_{u,k}}t_{k,1}^n\bm f_k(\mathbf q_{k,n})\right)^H\mathbf \Sigma_k\left(\sum_{m=1}^{N_\omega}s_1^m\bm g_k(\mathbf p_m)\right)\right|^4 = \left|\mathbf y_k^H\mathbf G_k\bm s_1\right|^4 = \left|\mathbf o_k^H\bm s_1\right|^4$. Accordingly, problem \eqref{P2_simp_sub1} can be equivalently expressed as 
\begin{eqnarray}\label{P_dis_sub1_eqv}
	\underset{\boldsymbol s_1}{\max}  \hspace{2mm} \sum_{k=1}^K\mu_k\left|\mathbf o_k^H\bm s_1\right|^4  \hspace{8mm}
	\text{s.t.} \hspace{2mm} \eqref{P2_simp_cons:c}, \eqref{P2_simp_cons:d}.
\end{eqnarray}
Then, the optimal solution can be easily determined as 
\begin{align}\label{equ:P2_simp_sub2_opt}
	\bm s_1^{\star} = \arg\underset{\boldsymbol s_1\in\mathcal S}{\max}\sum_{k=1}^K\mu_k\left|\mathbf o_k^H\bm s_1\right|^4,
\end{align}
where $\mathcal S\triangleq \{\mathbf e_1,\cdots,\mathbf e_{N_{\omega}}\}$, with $\mathbf e_m$ denoting the standard basis vector in $N_\omega$-dimensional space, having 1 at the $m$-th position and 0 at all other positions. For obtaining $\bm s_1^{\star}$, calculating $\left\lbrace\mathbf o_k\right\rbrace$ and $\left\lbrace \sum_{k=1}^K\mu_k\left|\mathbf o_k^H\bm s_1\right|^4 \Big|\boldsymbol s_1\in\mathcal S\right\rbrace$ have complexities of $\mathcal O\left(\sum_{k=1}^K\left(N_{u,k}\tilde L_k + \tilde L_k L_k + L_kN_\omega\right)\right)$ and $\mathcal O\left(KN_\omega^2\right)$, respectively. Thus, the overall complexity for acquiring $\bm s_1^{\star}$ is $\mathcal O\left(KN_\omega^2 + \sum_{k=1}^K\left(N_{u,k}\tilde L_k + \tilde L_k L_k + L_kN_\omega\right)\right)$. 

However, due to the special structure of $\boldsymbol s_1$, this complexity can be reduced. Specifically, we have   
\begin{align}\label{equ:trans_bina}
	&\sum_{k=1}^K\mu_k\left|\mathbf o_k^H\bm s_1\right|^4 \nonumber\\
	& \overset{(a)}{=} \sum_{k=1}^K\mu_k \left[\left| \left[\mathbf o_k^H\right]_1\right|^4, \left|\left[\mathbf o_k^H\right]_2\right|^4, \cdots, \left| \left[\mathbf o_k^H\right]_{N_\omega}\right|^4\right]\bm s_1 \nonumber\\
	& \triangleq \mathbf o^H\bm s_1, 
\end{align}
where the key equation $(a)$ holds because the vector $\bm s_1$ has exactly one element equal to 1 at an arbitrary position, with all other elements being 0, according to constraints \eqref{P2_simp_cons:c} and \eqref{P2_simp_cons:d}. In addition, the $m$-th element of $\mathbf o^H$ is given by $\sum_{k=1}^K\mu_k\left|\left[\mathbf o_k^H\right]_m\right|^4$. Then, the optimal solution for maximizing $\mathbf o^H\bm s_1$ is obtained as 
\begin{align}\label{equ:P2_simp_sub2_opt2}
	\bm s_1^{\star} = \mathbf e_{m^\star}, \ m^\star = \arg\underset{m\in\{1,2,\cdots,N_\omega\}}{\max}\left[\mathbf o^H\right]_m.
\end{align}
By utilizing this approach, the complexity of obtaining $\bm s_1^{\star}$ mainly depends on calculating $\left\lbrace\mathbf o_k\right\rbrace $ and $\bm o^H$, with complexities $\mathcal O\left(\sum_{k=1}^K\left(N_{u,k}\tilde L_k + \tilde L_k L_k + L_kN_\omega\right)\right)$ and $\mathcal O\left(KN_\omega\right)$, respectively. Given that $L_k \geq 1$, $\forall k\in\mathcal K$, the total complexity is $\mathcal O\left(\sum_{k=1}^K\left(N_{u,k}\tilde L_k + \tilde L_k L_k + L_kN_\omega\right)\right)$, which is smaller than $\mathcal O\left(KN_\omega^2 + \sum_{k=1}^K\left(N_{u,k}\tilde L_k + \tilde L_k L_k + L_kN_\omega\right)\right)$.    
 
\subsubsection{Optimizing $\{t_{k,1}^n\}$}
Given $\left\lbrace\tau_1,\{s_1^m\}\right\rbrace$, the optimization of $\{t_{k,1}^n\}$ can be executed separately and simultaneously for each $k\in\mathcal K$. Specifically, we define $\bm t_{k,1}^H\triangleq \left[t_{k,1}^1, t_{k,1}^2, \cdots, t_{k,1}^{N_{u,k}}\right] \in \mathbb R^{1\times N_{u,k}}$ and $\mathbf z_k \triangleq \mathbf F_k^H\mathbf x_k \in \mathbb C^{N_{u,k} \times 1}$ with $\mathbf x_k \triangleq \mathbf \Sigma_k\left( \sum_{m=1}^{N_\omega}s_1^m\bm g_k(\mathbf p_m)\right) \in \mathbb C^{\tilde L_k \times 1}$ and $\mathbf F_k = \left[\bm f_k(\mathbf q_{k,1}), \bm f_k(\mathbf q_{k,2}), \cdots, \bm f_k(\mathbf q_{k,N_{u,k}})\right] \in \mathbb C^{\tilde L_k \times N_{u,k}}$. Then, the subproblem with respect to $\bm t_{k,1}$ is given by 
\begin{subequations}\label{P2_simp_sub3}
	\begin{eqnarray}
		&\underset{\boldsymbol t_{k,1}}{\max}& \hspace{-2mm} \left|\bm t_{k,1}^H\mathbf z_k\right|^4\\
		&\text{s.t.}& \hspace{-2mm} t_{k,1}^n \in \{0,1\}, \ k\in\mathcal K, n \in\mathcal N_{u,k}, \label{P2_simp_sub3_cons:b}\\
		&& \hspace{-2mm} \sum_{n=1}^{N_{u,k}}t_{k,1}^n = 1, \ k\in\mathcal K, \label{P2_simp_sub_cons:c}
	\end{eqnarray}
\end{subequations}
whose optimal solution is straightforwardly obtained as
\begin{align}\label{equ:P2_simp_sub3_opt}
	\bm t_{k,1}^{\star} = \arg\underset{\boldsymbol t_{k,1}\in\mathcal T}{\max}\left|\bm t_{k,1}^H\mathbf z_k\right|^4,
\end{align}
where $\mathcal T\triangleq \{\mathbf e_1,\cdots,\mathbf e_{N_{u,k}}\}$. 
The complexity of calculating $\mathbf z_k$ and evaluating the set $\left\lbrace \left|\bm t_{k,1}^H\mathbf z_k\right|^4 \Big| \boldsymbol t_{k,1}\in\mathcal T \right\rbrace $ are $\mathcal O\left(N_\omega L_k + \tilde L_k L_k + N_{u,k}\tilde L_k\right)$ and $\mathcal O\left(N_{u,k}^2\right)$, respectively. Therefore, the total complexity of obtaining $\bm t_{k,1}^{\star}$ using \eqref{equ:P2_simp_sub3_opt} is $\mathcal O\left(N_{u,k}^2 + N_\omega L_k + \tilde L_k L_k + N_{u,k}\tilde L_k\right)$. Nevertheless, this complexity can be reduced similarly to the process of optimizing $\bm s_1$ by deriving that
\begin{align}
	\left|\bm t_{k,1}^H\mathbf z_k\right|^4 & = \bm t_{k,1}^H\left[\left|\left[\mathbf z_k\right]_1\right|^4, \left|\left[\mathbf z_k\right]_2\right|^4, \cdots, \left|\left[\mathbf z_k\right]_{N_{u,k}}\right|^4\right]^T \nonumber\\
	& \triangleq \bm t_{k,1}^H\hat {\mathbf z}_k.
\end{align}
Then, $\bm t_{k,1}^{\star}$ that maximizes $\bm t_{k,1}^H\hat {\mathbf z}_k$ can be obtained as
\begin{align}\label{equ:P2_simp_sub3_opt2}
	\bm t_{k,1}^{\star} = \mathbf e_{n^\star}, \ n^\star = \arg\underset{n\in\{1,2,\cdots,N_{u,k}\}}{\max}\left[\hat {\mathbf z}_k\right]_n.
\end{align}
In this way, the total complexity of obtaining $\bm t_{k,1}^{\star}$ is reduced to $\mathcal O\left(N_\omega L_k + \tilde L_k L_k + N_{u,k}\tilde L_k\right)$, since the complexities of calculating $ \mathbf z_k$, computing $\hat {\mathbf z}_k$, and finding $n^\star$ are $\mathcal O\left(N_\omega L_k + \tilde L_k L_k + N_{u,k}\tilde L_k\right)$, $\mathcal O\left(N_{u,k}\right)$, and $\mathcal O\left(N_{u,k}\right)$, respectively. 

\subsubsection{Optimizing $\tau_1$} 
Given other variables, the subproblem for optimizing $\tau_1$ has a similar form to problem \eqref{P1_simp_sub3}, and the expression for its optimal solution can be obtained by replacing ``$c$'' in \eqref{eq:P1_simp_sub3_solution} with ``$\hat c\triangleq \frac{\sum_{k=1}^K{\zeta_kP_{\rm A}\left|h_{k,1}({s_1^m},{t_{k,1}^n})\right|^4}}{\sigma^2}$''.

\subsubsection{Convergence and Complexity Analysis}
Similar to the analysis in the previous section, the proposed algorithm guarantees convergence to a suboptimal solution by iteratively updating one of the three variable subsets while keeping the others fixed. Additionally, based on the results presented earlier, the total complexity of each iteration of this algorithm is about $\mathcal O\left(\sum_{k=1}^K\left(N_{u,k}\tilde L_k + \tilde L_k L_k + L_kN_\omega\right)\right)$. 

\section{Simulation Results}\label{Sec:simu}
In the simulation setup, the system operates at a frequency of 5 GHz, corresponding to a wavelength of $\lambda = 0.06$ m \cite{2023_Yifei_discrete,2024_Weidong_graph}. The HAP is positioned at $[0,0]^T$ and the WDs are distributed randomly within a 1.5-meter disk centered at $[10,0]^T$. The geometric channel model is adopted, where $L_k = \tilde L_k \triangleq L$, $\forall k \in \mathcal{K}$. As a result, each WD's path-response matrix is diagonal, represented as $\mathbf \Sigma_k = {\rm diag}\{\sigma_{k,1},\cdots,\sigma_{k,L}\}$, where $\sigma_{k,\ell}$ follows a distribution $\mathcal{CN}\left(0,c_k^2/L\right)$, $\ell = 1,\dots,L$ \cite{2023_Lipeng_Modeling}. The term $c_k^2$ is defined as $C_0D_k^{-\alpha}$, where $C_0 = \left({\lambda}/{4\pi}\right)^2$ is the expected average channel power gain at the reference distance of 1 m, $D_k$ is the distance between WD $k$ and the HAP, and $\alpha = 2.8$ is the path-loss exponent \cite{2023_Lipeng_uplink}. Both elevation and azimuth AoDs/AoAs are assumed to be uniformly distributed within the interval $\left[0, \pi\right]$ \cite{2023_Wenyan_MIMO}. The movement regions for the MAs are defined as $\mathcal C_{\omega} = \mathcal C_k = \left[-A/2, A/2\right]\times \left[-A/2, A/2\right]$, $\forall k \in \mathcal{K}$, with $A$ intentionally limited to $A \leq 8\lambda$ to ensure far-field conditions. The Rayleigh distance corresponding to $A = 8\lambda$ is $\frac{2\times(8\lambda)^2}{\lambda} = 7.68$ m, which is smaller than the minimum separation of 8.5 m between the HAP and any WD, thereby validating the far-field assumption. 
Moreover, given the identical step size $d$ for each MA in discrete antenna positioning, we have $N_\omega = N_{u,k} \triangleq N$, $\forall k \in \mathcal K$, where $N$ is determined by the values of $A$ and $d$. 
Unless specified otherwise, the remaining parameters are set as follows: $P_{\max} = 40$ dBm, $T = 3$ s, $K = 5$, $L = 10$, $A = 5\lambda$, $d = 1/4\lambda$, $v = 0.125$ m/s \cite{2022_Xingxing_MAspeed}, and $\sigma^2 = -90$ dBm \cite{2024_Jinze_full-duplex}. Each simulation result is obtained by averaging over 500 independent realizations with randomly generated user locations and channel realizations.


For comparison, we consider the following schemes: 
\begin{itemize}
	\item \textbf{Proposed continuous MA:} the approach in Section \ref{Sec:P1_simp_solu}. 
	\item \textbf{Proposed discrete MA:} the approach in Section \ref{Sec:P2_simp_solu}. 
	\item \textbf{Partially MA:} in each channel realization, only $(K+1)/2$ antennas are randomly selected for free movement, while the other $(K+1)/2$ antennas remain fixed at the reference points within their respective moving regions. The positions of the $(K+1)/2$ MAs are jointly optimized along with the time allocation. 
	\item \textbf{Random MA position:} in each channel realization, we randomly and independently generate 500 samples of $\left\lbrace\{s_1^m\},\{t_{k,1}^n\}\right\rbrace$, ensuring each of them satisfies the constraints in \eqref{P2_simp_cons:c}-\eqref{P2_simp_cons:f}. For each sample, the time allocation is optimized. The highest-performing solution from these 500 samples is chosen as the final output of this scheme. 
	\item \textbf{FPA w/ compensation time:} the time allocation is optimized with all $K+1$ antennas configured as FPAs. To ensure a fair comparison with the continuous MA scheme, this baseline incorporates an additional compensation transmission time, denoted by $\tau_0(v)$, which corresponds to the initial movement duration required by the proposed continuous MA scheme. The movements are initialized from the reference points defined in Section \ref{Sec:model_and_formulation}, located within the designated moving regions. These reference points also serve as the locations of the FPAs, thereby enabling a consistent comparison between the MA and FPA schemes. Similar to \eqref{eq:tau_2}, the value of $\tau_0(v)$ is given by 
	\begin{align}
		\tau_0(v) = \max\left(\frac{\left\|\boldsymbol \omega_1 - \boldsymbol \omega_0\right\|_1}{v}, \max_{k\in\mathcal K}\frac{\left\| \boldsymbol u_{k,1} - \boldsymbol u_{k,0}\right\|_1}{v}\right), 
	\end{align}
	and the total transmission time becomes $T + \tau_0(v)$. 
	\item \textbf{FPA w/o compensation time:} the time allocation is optimized with all $K+1$ antennas being FPAs. The total transmission time remains $T$, consistent with the proposed two MA schemes.  
\end{itemize}

\begin{figure}[!t]
	\centering
	\includegraphics[scale=0.68]{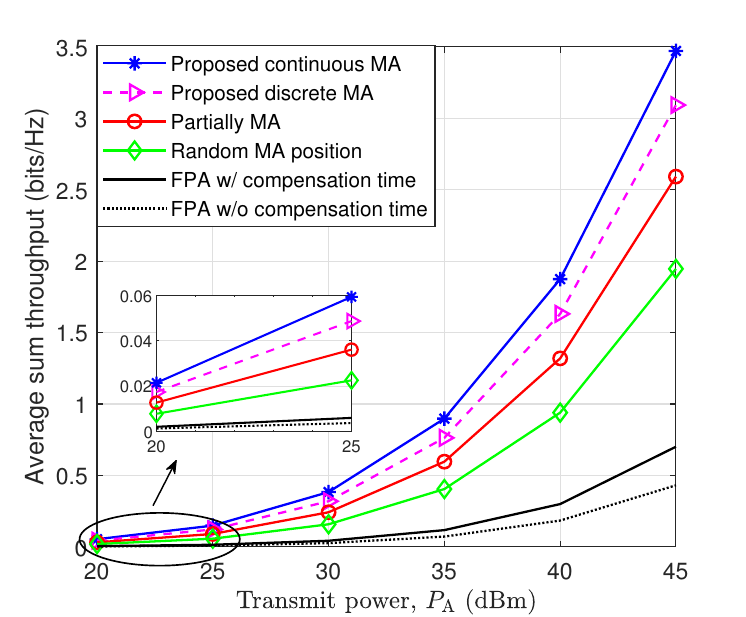}
	\caption{Average sum throughput versus transmit power at the HAP.}
	\label{fig:TP_vs_PA}
	\vspace{-2mm}
\end{figure}

Fig. \ref{fig:TP_vs_PA} plots the system sum throughput obtained by different schemes versus the transmit power $P_{\rm A}$ at the HAP. It is observed that the four MA-based schemes significantly improve the sum throughput compared to the two FPA-based schemes, even though one of the FPA schemes includes compensation transmission time. This improvement is attributed to the strategic placement of the MAs, which enhances the channel conditions between the HAP and WDs, resulting in greater efficiency in both downlink WPT and uplink WIT. In particular, the proposed continuous MA scheme consistently achieves the highest sum throughput, demonstrating performance improvements of approximately 12.30\%, 33.98\%, 77.95\%, 395.71\%, and 706.98\% over the discrete MA, partially MA, random MA position, FPA with compensation time, and FPA without compensation time schemes, respectively.  This is expected, as the proposed continuous MA scheme exploits the most spatial degrees of freedom, and the other three MA-based schemes suffer performance loss due to less flexibility in channel reconfiguration. 

\begin{figure*}[!t]
	\subfigure[]{\label{fig:Tau_vs_PA}
		\includegraphics[scale=0.68]{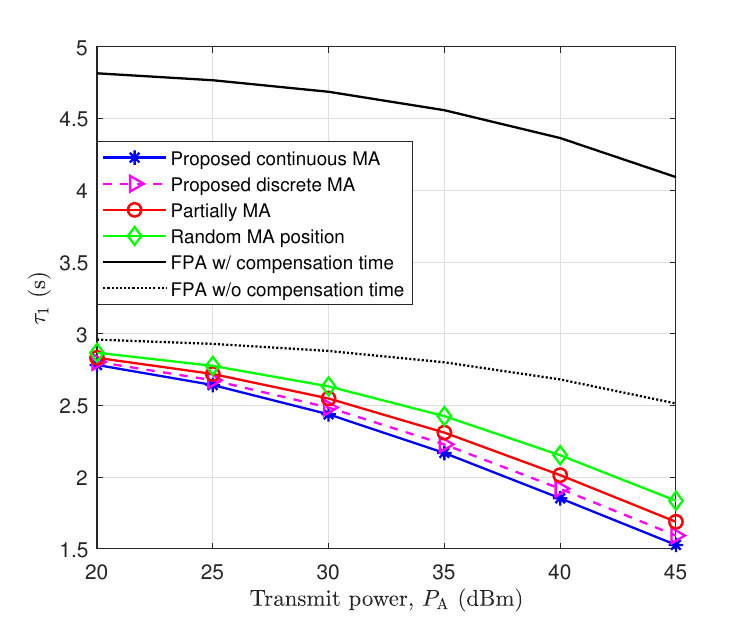}}
	\hspace{4mm}
	\subfigure[]{\label{fig:Energy_vs_PA}
		\includegraphics[scale=0.68]{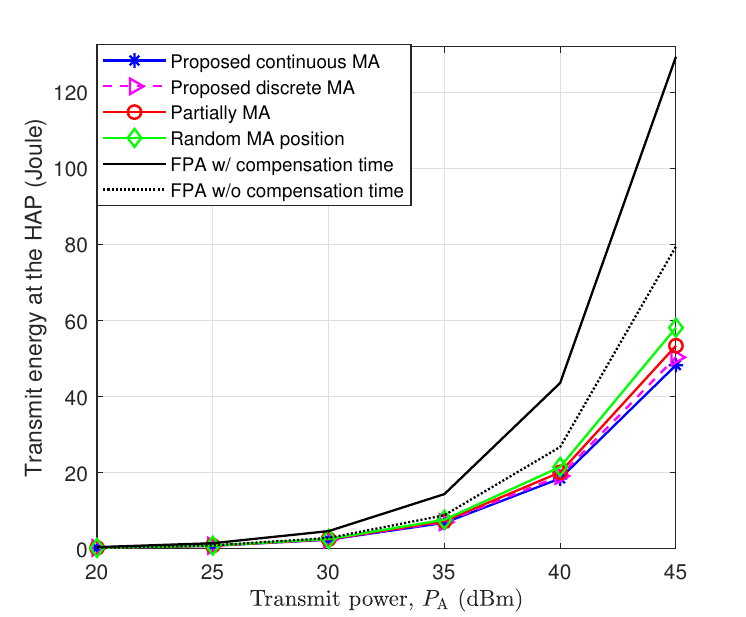}}
	\caption{(a) Average downlink WPT duration and (b) average total energy consumed at the HAP versus maximum transmit power at the HAP. }
	\label{fig:R_vs_IUs}
\end{figure*}

\begin{figure*}[!t]
	\subfigure[]{\label{fig:TP_vs_T}
		\includegraphics[scale=0.68]{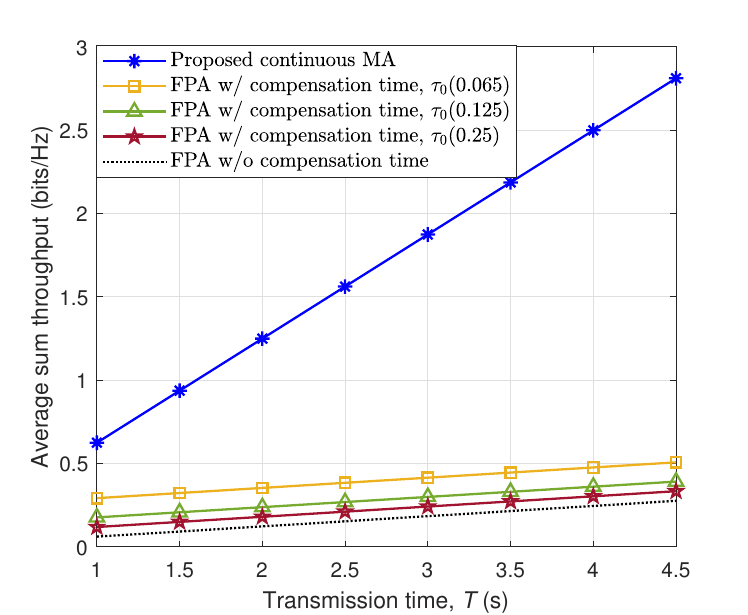}}
	\hspace{4mm}
	\subfigure[]{\label{fig:T_ratio_movT}
		\includegraphics[scale=0.68]{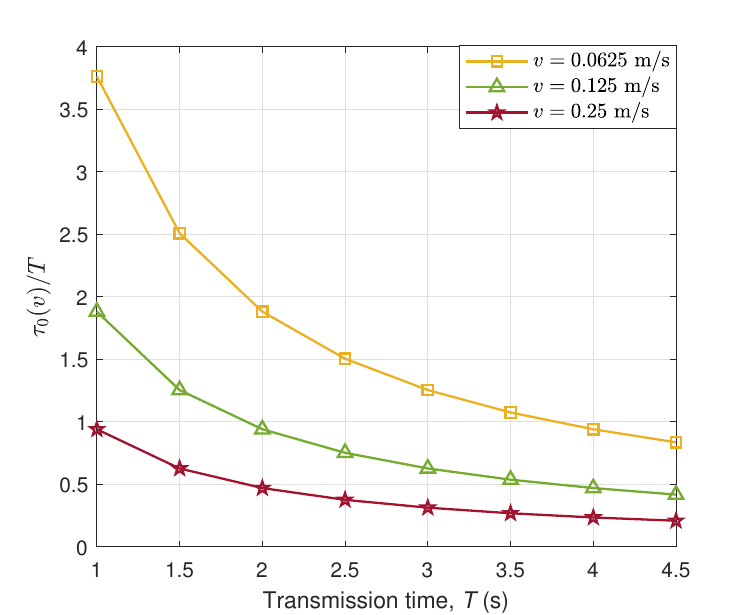}}
	\caption{(a) Average sum throughput and (b) ratio of initial movement duration to  transmission time, $\tau_0(v)/T$, under different antenna-moving speeds, both plotted versus transmission time $T$.}
	\label{fig:TP_vs_T_ab}
	\vspace{-2mm}
\end{figure*}

\begin{figure}[!h]
	\centering
	\includegraphics[scale=0.68]{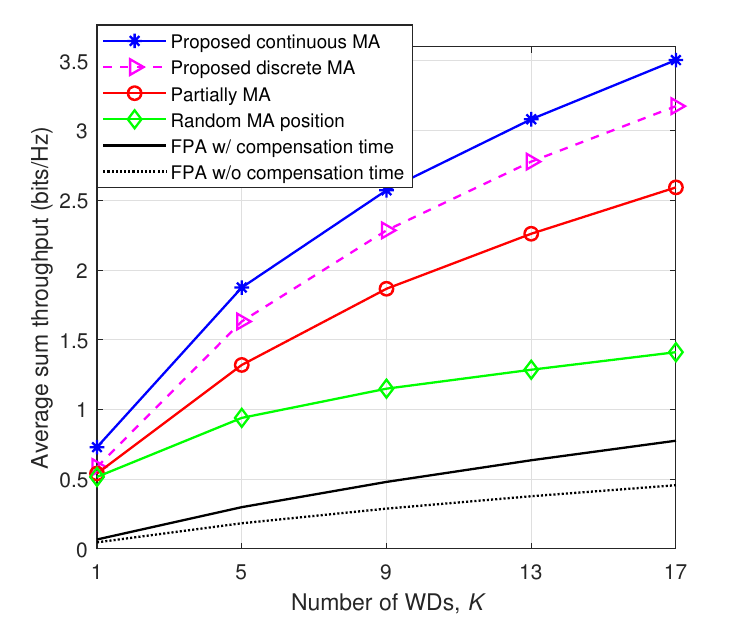}
	\caption{Average sum throughput versus number of WDs.}
	\label{fig:TP_vs_K}
	\vspace{-1mm}
\end{figure}

\begin{figure}[!t]
	\subfigure[$L = 10$.]{\label{fig:TP_vs_reg_L10}
		\includegraphics[scale=0.68]{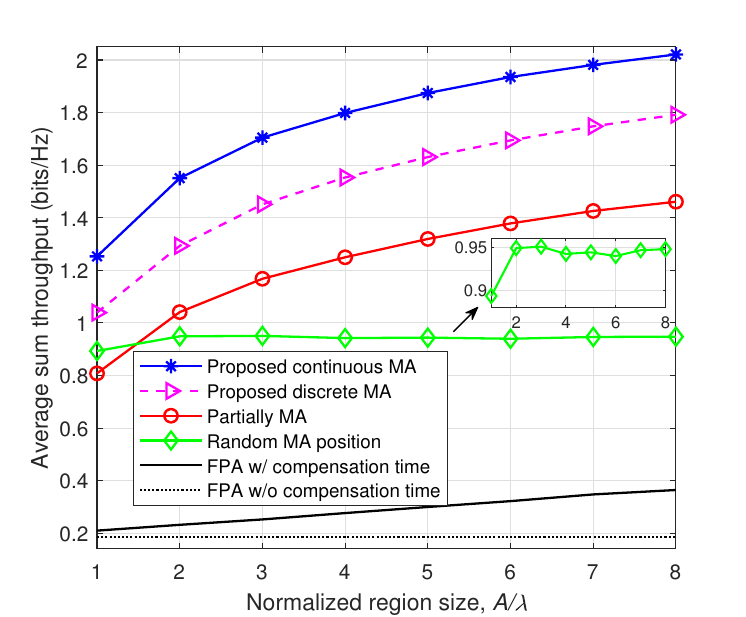}}
	\subfigure[$L = 4$.]{\label{fig:TP_vs_reg_L4}
		\includegraphics[scale=0.68]{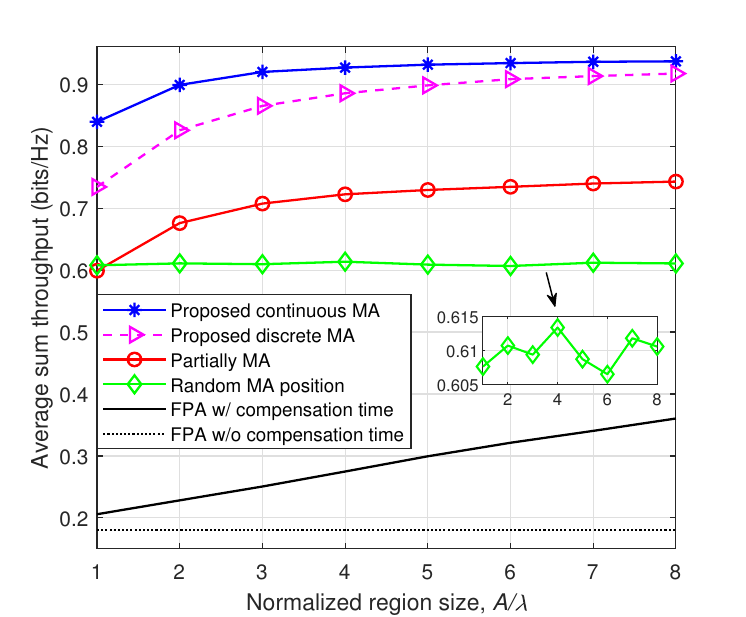}}
	\caption{Average sum throughput versus normalized region size.}
	\label{fig:TP_vs_reg}
	\vspace{-3mm}
\end{figure}

To gain further insights, we plot the corresponding optimized downlink WPT duration and total energy consumption at the HAP, given by $E_{\rm HAP} = P_{\rm A}\tau_1$, versus $P_{\rm A}$ in Figs. \ref{fig:Tau_vs_PA} and \ref{fig:Energy_vs_PA}, respectively. Fig. \ref{fig:Tau_vs_PA} implies that for any value of $P_{\rm A}$, the FPA with compensation time scheme experiences a significantly longer transmission time, resulting in a downlink WPT duration that is considerably greater than those of the other schemes. Even so, the sum throughput achieved by this scheme is far lower than that of the proposed continuous and discrete MA schemes (see Fig. \ref{fig:TP_vs_PA}). This demonstrates that improving the channel conditions is more effective than extending the transmission duration in enhancing the system sum throughput. Moreover, we observe from Fig. \ref{fig:Energy_vs_PA} that the continuous MA scheme consumes the least transmit energy at the HAP, followed by the discrete MA scheme, while the FPA with compensation time scheme consumes the most, more than twice that of the continuous MA scheme. This further highlights the disadvantages of this FPA scheme. 

In Fig. \ref{fig:TP_vs_T}, we examine the impact of transmission time on the performance of the proposed continuous MA scheme and the two FPA-based schemes, considering three different values of $v$: $0.0625$, $0.125$, and $0.25$ m/s \cite{2022_Xingxing_MAspeed}. Note that as $v$ increases, the compensation time $\tau_0(v)$ decreases, leading to reduced transmission duration and hence lower performance for the FPA with compensation time scheme. Moreover, The performance gap between the continuous MA scheme and the FPA-based schemes increases with the transmission time $T$, while the spacing among different FPA variants remains nearly unchanged. This is because both types of schemes achieve throughput that scales linearly with $T$, but with different growth rates. Specifically, define $A \triangleq \left(1 -\gamma_1^{\star}\right) \log_2\left(1 + \frac{c^{\star}\gamma_1^{\star}}{1-\gamma_1^{\star}}\right)$, $\gamma_1^\star \triangleq \frac{\tau_1^{\star}}{T} = \frac{\left(\exp\left(\mathcal W\left(\frac{c^{\star}-1}{e}\right) + 1\right)-1\right) }{c^{\star}+\exp\left(\mathcal W\left(\frac{c^{\star}-1}{e}\right) + 1\right)-1}$, and $c^{\star}\triangleq \frac{\sum_{k=1}^K{\zeta_kP_{\rm A}\left|h_{k,1}(\boldsymbol \omega_1^{\star}, \boldsymbol u_{k,1}^{\star})\right|^4}}{\sigma^2}$. The throughput of the continuous MA scheme is then given by $R_{\rm MA} = TA$. For an FPA-based scheme with compensation time $\tau_0(v)$, the  throughput is $R_{\rm FPA} = \left(T+\tau_0(v)\right)B = TB + \tau_0(v)B$, where $B \triangleq \left(1-\gamma_2^{\star} \right) \log_2\left(1 + \frac{c_0\gamma_2^{\star}}{1-\gamma_2^{\star}}\right)$, $\gamma_2^{\star} \triangleq \frac{\tau_1^{\star}}{T+\tau_0(v)} = \frac{\left(\exp\left(\mathcal W\left(\frac{c_0-1}{e}\right) + 1\right)-1\right) }{c+\exp\left(\mathcal W\left(\frac{c_0-1}{e}\right) + 1\right)-1}$, and $c_0\triangleq \frac{\sum_{k=1}^K{\zeta_kP_{\rm A}\left|h_{k,1}(\boldsymbol \omega_0, \boldsymbol u_{k,0})\right|^4}}{\sigma^2}$. Since $A > B$, the gap between the continuous MA and FPA-based schemes grows linearly with $T$, following $\Delta R(T) = R_{\rm MA} - R_{\rm FPA} = T(A-B)-\tau_0(v)B$. In contrast, the difference between any two FPA variants with compensation times $\tau_0(v_1)$ and $\tau_0(v_2)$ is $R_{\rm FPA}^{(v_1)} - R_{\rm FPA}^{(v_2)} = B\left[\tau_0(v_1)-\tau_0(v_2)\right]$, which is independent of $T$. Therefore, the gap between the continuous MA scheme and the FPA-based schemes widens as $T$ increases, while the spacing among FPA variants remains nearly constant. 

In Fig. \ref{fig:T_ratio_movT}, we plot the corresponding ratio of initial movement duration to transmission time, $\tau_0(v)/T$, under different antenna-moving speeds. As shown, this ratio decreases with increasing $T$, and higher speeds leads to smaller values. When $T$ is sufficiently large $T$ and the antenna-moving speed is high, the initial movement duration $\tau_0(v)$, which occurs during the preprocessing stage, incurs only a minor overhead relative to the subsequent transmission phase. Moreover, as illustrated in Figs. \ref{fig:TP_vs_T} and \ref{fig:T_ratio_movT}, even when both $T$ and $v$ are small (e.g., $T = 1$ s and $v = 0.0625$ m/s), leading to $\tau_0(v)$ being approximately 3.7 times longer than $T$, the FPA with compensation time scheme still achieves a lower throughput than the proposed continuous MA scheme. This result underscores that the performance advantage of MA repositioning remains significant even when the associated preprocessing overhead is relatively large.

Fig. \ref{fig:TP_vs_K} depicts the sum throughput versus the number of WDs. Both the proposed continuous and discrete MA schemes consistently outperform all baselines, with their performance advantage becoming more pronounced as $K$ increases. This trend arises because each WD is equipped with a dedicated MA, and a larger $K$ naturally leads to more MAs in the system. In the random MA position scheme, each MA is placed randomly without adapting to user-specific channel conditions. As $K$ grows, the cumulative impact of these non-optimized placements becomes more pronounced, resulting in greater performance degradation. In contrast, the continuous and discrete MA schemes optimize each MA's position based on the corresponding channel environment, thereby enabling more favorable propagation and improved spatial resource utilization. This adaptive positioning ensures strong scalability with $K$. Moreover, as $K$ increases, the probability of a larger $\tau_0(v)$ also increases, thereby enlarging the performance gap between the two FPA-based schemes due to the extended transmission time in the compensation scheme. 

In Fig. \ref{fig:TP_vs_reg}, we plot the system sum throughput versus the normalized region size $A/\lambda$ when $L = 10$ and $L = 4$, respectively. Firstly, it is observed that the sum throughput of the continuous MA, discrete MA, and partially MA schemes increase with $A/\lambda$, due to the greater flexibility in antenna movement within larger regions, allowing for more effective channel reconfiguration and improved transmission efficiency. However, while the throughput of these schemes tends to saturate as $A/\lambda$ increases for $L = 4$, noticeable growth is observed for $L = 10$. This is consistent with the theoretical and numerical results in \cite{2023_Lipeng_Modeling}, which suggest that to achieve maximum performance, a larger moving region is required when more channel paths are involved. 
Secondly, the performance gap between the continuous MA and discrete MA schemes narrows more noticeably with increasing $A/\lambda$ for $L = 4$. The reasons are twofold. For one thing, as $A/\lambda$ increases, the sum throughput of the continuous MA scheme saturates, and the periodic nature of the channel gains results in more optimal positions that achieve maximum performance. For another, as the number of candidate positions for the discrete MA scheme increases with $A/\lambda$, the likelihood of including some of these optimal positions grows, enhancing the chances of achieving performance comparable to the continuous MA scheme. 
Thirdly, the performance of the random MA position scheme exhibits relatively flat behavior across different $A/\lambda$. This is due to the lack of optimization in antenna placement, as random positioning cannot consistently exploit favorable channel conditions, resulting in suboptimal performance. Finally, in Fig. \ref{fig:TP_vs_reg_L4}, for the MA-based schemes, the throughput tends to saturate as the region size increases, since the superimposed channel gain exhibits a periodic variation with respect to antenna displacement, and the maximum performance can typically be achieved within a finite moving region \cite{2023_Lipeng_Modeling}. However, as the region size increases, the optimized positions derived from the continuous MA scheme are more likely to be located farther from the initial positions, resulting in a larger $\tau_0(v)$. Since the throughput of the FPA with compensation time scheme is positively correlated with the total transmission duration $T + \tau_0(v)$, its performance continues to improve even when the other schemes have already saturated.

\begin{figure}[!t]
	\centering
	\includegraphics[scale=0.68]{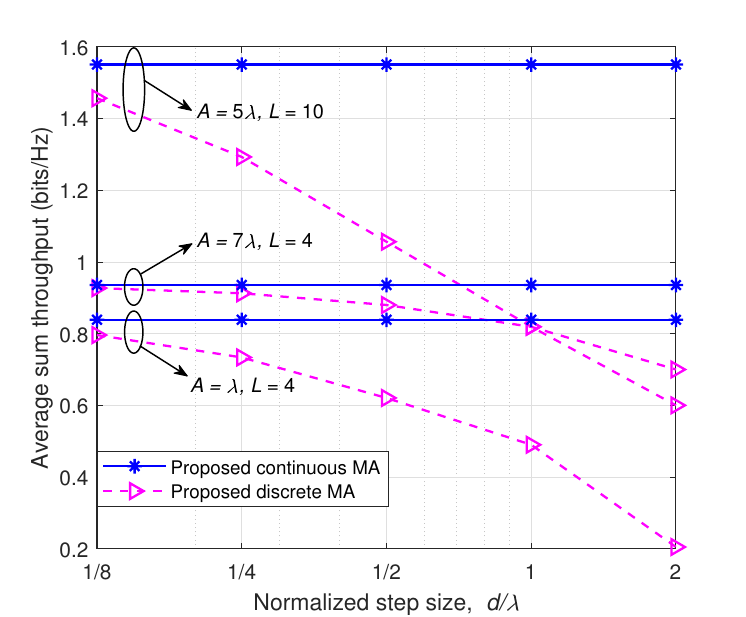}
	\caption{Average sum throughput versus normalized step size.}
	\label{fig:TP_vs_d}
	\vspace{-4mm}
\end{figure}

In Fig. \ref{fig:TP_vs_d}, the impact of the normalized step size, $d/\lambda$, on the performance of the proposed discrete MA scheme is investigated. It is observed that the sum throughput of the discrete MA scheme decreases as $d/\lambda$ increases. This is expected, as a larger $d/\lambda$ results in fewer candidate positions to select from, reducing flexibility in channel reconfiguration and consequently leading to lower performance. We also note that, for all three different values of $A$ and $L$, a normalized step size of $d/\lambda = 1/4$ is sufficient for the discrete MA scheme to achieve over 80\% of the performance of the continuous MA scheme. Moreover, for $A = 7\lambda$ and $L = 4$ (where the region size is large enough to include multiple optimal positions that achieve maximum performance, as discussed in the previous paragraph), the discrete MA scheme performs comparably to the continuous MA scheme, even with $d/\lambda = 1/2$.

\section{Conclusions}\label{Sec:conclu}
This paper studied an MA-aided WPCN utilizing NOMA for uplink WIT, offering a key distinction from traditional FPA-based WPCNs by enabling the MAs at both the HAP and WDs to adjust their positions before downlink WPT and uplink WIT. We considered two antenna movement patterns: continuous and discrete. To maximize the system sum throughput, we formulated two design problems corresponding to these movement patterns, where the MA positions, the time allocation, and the uplink power allocation were jointly optimized. To address these two non-convex optimization problems, we first revealed that the optimum for each is achieved using identical MA positions for both downlink WPT and uplink WIT. Building on this result, we developed computationally efficient algorithms using AO, where the optimization variables are split into three blocks for easier handling. Particularly, in the discrete positioning scenario, each subproblem was solved optimally with acceptable computational complexity. Numerical results demonstrated that the proposed designs can significantly boost system sum throughput compared to several baseline schemes. Moreover, key insights were gained regarding the performance comparison between the discrete and continuous MA schemes. Specifically, the discrete MA scheme was found to achieve a significant portion of the continuous MA scheme's throughput with a moderate step size, and when each antenna moving region was sufficiently large, it delivered comparable performance without requiring a tiny step size.  

This paper considered far-field channel conditions. However, the far-field assumption may no longer hold as the Rayleigh distance increases with the movement range or carrier frequency. This highlights the need to investigate near-field propagation scenarios, as well as hybrid environments where both near- and far-field effects coexist.  In addition, it is of interest to explore more general settings that account for imperfect CSI, multi-antenna WDs, and hardware impairments at both the HAP and the WDs. These extensions introduce new challenges in problem formulation and algorithm design, which are left for future work.

\appendix[Proof of Proposition \ref{prop1}]
Before proving Proposition \ref{prop1}, we first introduce the following lemma.
\begin{lem}\label{lem1}
	\rm For any $a_k\geq 0$ and $b_k \geq 0$, $k \in \mathcal K$, it holds that 
	\begin{align}\label{ineq:lem}
		\log_2\left( 1+\sum_{k=1}^Ka_kb_k\right) \leq & \log_2\left( \sqrt{1+\sum_{k=1}^Ka_k^2}\right) \nonumber\\
		& + \log_2\left(\sqrt{1+\sum_{k=1}^Kb_k^2}\right).
	\end{align}
\end{lem}
\begin{proof}
	First, we define $a_0 = b_0 = 1$. By the Cauchy–Schwarz inequality, we have $\left( \sum_{k=0}^Ka_kb_k\right)^2 \leq \left(\sum_{k=0}^Ka_k^2\right) \left(\sum_{k=0}^Kb_k^2\right)$, which implies $1+\sum_{k=1}^Ka_kb_k \leq \sqrt{\left(1+\sum_{k=1}^Ka_k^2\right)\left( 1+\sum_{k=1}^Kb_k^2\right)}$. Furthermore, since $\log_2(\cdot)$ is a monotonically increasing function, we have 
	\begin{align}\label{ineq:lem_fuzhu}
		& \log_2\left( 1+\sum_{k=1}^Ka_kb_k\right) \nonumber\\
		& \hspace{1cm}\leq \log_2\left( \sqrt{\left(1+\sum_{k=1}^Ka_k^2\right)\left( 1+\sum_{k=1}^Kb_k^2\right)}\right). 
	\end{align}
	Applying the product rule of logarithms to \eqref{ineq:lem_fuzhu}, we arrive at \eqref{ineq:lem}. This completes the proof of Lemma \ref{lem1}. 
\end{proof}

We are now ready to prove Proposition \ref{prop1}. First, at the optimal solution to (P1), the constraints in \eqref{P1_cons:d} must be active, i.e.,  $p_k^* = \frac{\zeta P_{\rm A}\left|h_{k,1}(\boldsymbol \omega_1^*, \boldsymbol u_{k,1}^*)\right|^2\tau_1^*}{\tau_3^*}$, $\forall k\in\mathcal K$. This can be proved by contradiction: if any constraint in \eqref{P1_cons:d} holds with strict inequality at the optimum, the objective value of (P1) can be further improved by increasing the corresponding $p_k$ until the inequality becomes an equality. Then, by substituting the expression of $p_k^*$ into the objective function of (P1), (P1) becomes  
\begin{eqnarray}\label{P1_eqv}
	\underset{\mathcal Z_1\setminus\{p_k\}}{\max} \hspace{2mm} \hat R_{\rm sum}^{\rm cont} \hspace{8mm}
	\text{s.t.} \hspace{2mm} \eqref{P1_cons:b}, \eqref{P1_cons:c}, \eqref{P1_cons:e}, \eqref{P1_cons:f}, 
\end{eqnarray}
where 
\begin{align}
	 & \hat R_{\rm sum}^{\rm cont} = \tau_3 \nonumber\\ 
	 & \!\times\! \log_2\left(1 \!+\! \sum_{k=1}^K\frac{\zeta P_{\rm A}\tau_1\left|h_{k,1}(\boldsymbol \omega_1, \boldsymbol u_{k,1})\right|^2\left|h_{k,2}(\boldsymbol \omega_2, \boldsymbol u_{k,2})\right|^2}{\sigma^2\tau_3}\right).
\end{align}
Next, we define $z_k\left(\boldsymbol \omega,\boldsymbol u_k\right) \triangleq \frac{\sqrt{\zeta P_{\rm A}\tau_1}\left|\bm f_k(\boldsymbol u_k)^H\mathbf \Sigma_k\bm g_k(\boldsymbol \omega)\right|^2}{\sqrt{\sigma^2\tau_3}}$ and $z_k^* \triangleq \underset{\boldsymbol \omega \in\mathcal C_\omega,\boldsymbol u_k \in\mathcal C_{u,k}}{\max}z_k\left(\boldsymbol \omega,\boldsymbol u_k\right)$. Then, $\hat R_{\rm sum}^{\rm cont}$ can be rewritten as $\tau_3\log_2\left(1 + \sum_{k=1}^Kz_k\left(\boldsymbol \omega_1,\boldsymbol u_{k,1}\right)z_k\left(\boldsymbol \omega_2,\boldsymbol u_{k,2}\right)\right)$, which satisfies the following inequality: 
\begin{align}
	&\tau_3\log_2\left(1 + \sum_{k=1}^Kz_k\left(\boldsymbol \omega_1,\boldsymbol u_{k,1}\right)z_k\left(\boldsymbol \omega_2,\boldsymbol u_{k,2}\right)\right) \nonumber\\
	& \overset{(a)}{\leq} \tau_3\log_2\left(\sqrt{1+\sum_{k=1}^Kz_k^2\left(\boldsymbol \omega_1,\boldsymbol u_{k,1}\right)}\right) \nonumber\\
	& \hspace{6mm} + \tau_3\log_2\left(\sqrt{1+\sum_{k=1}^Kz_k^2\left(\boldsymbol \omega_2,\boldsymbol u_{k,2}\right)}\right) \nonumber\\
	& \overset{(b)}{\leq} \tau_3\log_2\left(\sqrt{1+\sum_{k=1}^K\left( z_k^*\right) ^2}\right),
\end{align}
where $(a)$ follows from Lemma \ref{lem1}, $(b)$ holds because $z_k^* \triangleq \underset{\boldsymbol \omega \in\mathcal C_\omega,\boldsymbol u_k \in\mathcal C_{u,k}}{\max}z_k\left(\boldsymbol \omega,\boldsymbol u_k\right)$, and the equality in $(b)$ holds when $z_k\left(\boldsymbol \omega_1^*,\boldsymbol u_{k,1}^*\right) = z_k\left(\boldsymbol \omega_2^*,\boldsymbol u_{k,2}^*\right) = z_k^*$. Since the function $z_k\left(\boldsymbol \omega,\boldsymbol u_k\right)$  exhibits periodic behavior due to the presence of the cosine function in its expansion, there may exist distinct solutions $\boldsymbol \omega_1^* \neq \boldsymbol \omega_2^*$ and $\boldsymbol u_{k,1}^* \neq \boldsymbol u_{k,2}^*$ for some $k\in\mathcal K$. However, if $\boldsymbol \omega_1^* = \boldsymbol \omega_2^*$ and $\boldsymbol u_{k,1}^* = \boldsymbol u_{k,2}^*$, $\forall k\in\mathcal K$, we have $\tau_2^* = \max\left(\frac{\left\|\boldsymbol \omega_2^* - \boldsymbol \omega_1^*\right\|_1}{v}, \max_{k\in\mathcal K}\frac{\left\| \boldsymbol u_{k,2}^* - \boldsymbol u_{k,1}^*\right\|_1}{v}\right) = 0$. Otherwise, $\tau_2^* > 0$. Clearly, $\tau_2^* > 0$ cannot be the optimal solution to (P1), as the positive duration $\tau_2$ can be reallocated to the downlink WPT phase, resulting in higher $\{p_k^*\}$ and, consequently, an increase in the system sum throughput. Therefore, we have $\boldsymbol \omega_1^* = \boldsymbol \omega_2^*$, $\boldsymbol u_{k,1}^* = \boldsymbol u_{k,2}^*$, $\forall k\in\mathcal K$, and $\tau_2^* = 0$. Then, the objective function becomes $\tau_3^*\log_2\left(1 + \sum_{k=1}^K\frac{\zeta P_{\rm A}\tau_1^*\left|h_{k,1}(\boldsymbol \omega_1^*, \boldsymbol u_{k,1}^*)\right|^4}{\sigma^2\tau_3^*}\right)$, and constraint \eqref{P1_cons:b} becomes $\tau_1^* + \tau_3^* \leq T$. It is easy to see that at the optimal solution of (P1), there must be $\tau_1^* + \tau_3^* = T$ since otherwise the rest of $T$ can be allocated to $\tau_1$ to further improve the objective value. Combining the above results completes the proof.

\bibliographystyle{IEEEtran}
\bibliography{ref}

\end{document}